\documentclass[a4paper,12pt]{report}
\usepackage{amsthm}
\usepackage{amsmath}
\usepackage{float}
\usepackage{hyperref}
\hypersetup{
	  colorlinks   = true, 
	  urlcolor     = blue, 
	  linkcolor    = blue, 
	  citecolor   = red 
	}
\usepackage{amssymb}
\usepackage{amsthm}
\usepackage{verbatim}
\usepackage{psfrag}
\usepackage{fancyhdr}
 \usepackage{listings}
 \usepackage[chapter]{algorithm}
 \usepackage{algpseudocode}
\ifx\pdftexversion\undefined
\usepackage[dvips]{graphicx}
\else
\usepackage{graphicx}
\fi
\usepackage{epstopdf}
\usepackage{float}
\usepackage{subfig}

\newpage
\renewcommand{\headrulewidth}{0pt}

\DeclareMathOperator{\AGL}{AGL}
\DeclareMathOperator{\PGL}{PGL}
\DeclareMathOperator{\PGAL}{P\Gamma L}
\newcommand{\Sym}[1]{\operatorname{S}_{#1}}
\newtheorem{theorem}{Theorem}[chapter]

\newtheorem{definition}[theorem]{Definition}
\newtheorem{example}[theorem]{Example}

\newtheorem{lemma}[theorem]{Lemma}

\theoremstyle{definition}
\include{commands}

\begin{document}
\pagenumbering{roman}
\fancyfoot[CO]{\thepage}
 \cleardoublepage
 \title{Deterministic Polynomial Factoring Under The Assumption of Riemann Hypothesis}
 \newpage
 \lstset{language=[ANSI]C}
 \lstset{
 basicstyle=\footnotesize\tt, 
 identifierstyle=, 
 commentstyle=\color{blue}, 
 showstringspaces=false, 
 lineskip=1pt,
 captionpos=b,
 frame=single,
 breaklines=true
 }
 \lstset{classoffset=0,
 morekeywords={},keywordstyle=\color{black},
 classoffset=1,
 classoffset=0}
 \author{Aurko Roy}

\begin{center}
\noindent {\bf{\Large{Deterministic Polynomial Factoring Under The Assumption of Extended Riemann Hypothesis}}}\\

\vspace{25mm}

A Thesis Submitted\\
in Partial Fulfilment of the Requirements
\\
for the Degree of
\\
\vspace{30mm}
{\bf {\large MASTER OF TECHNOLOGY}}
\\
by
\\
{\bf{\Large{Aurko Roy}}}\\
\vspace{20mm}
\begin{figure}[h]
\centering
\includegraphics[width=0.25\textwidth]{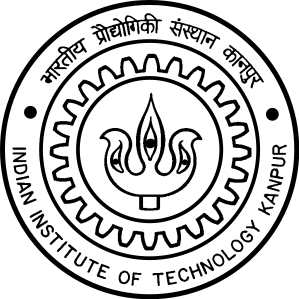}
\end{figure}

\vspace{3mm}
{\bf {\large {\sc DEPARTMENT OF COMPUTER SCIENCE AND ENGINEERING}}}\\
\vspace{2mm}
{\bf {\large {\sc INDIAN INSTITUTE OF TECHNOLOGY KANPUR}}}\\
\vspace{3mm}
{\textbf{ December, 2012}}\\
\end{center}

\newpage
\setcounter{page}{2}
\addcontentsline{toc}{chapter}{Certificate}
\includegraphics[width=\textwidth,height=\textheight,keepaspectratio]{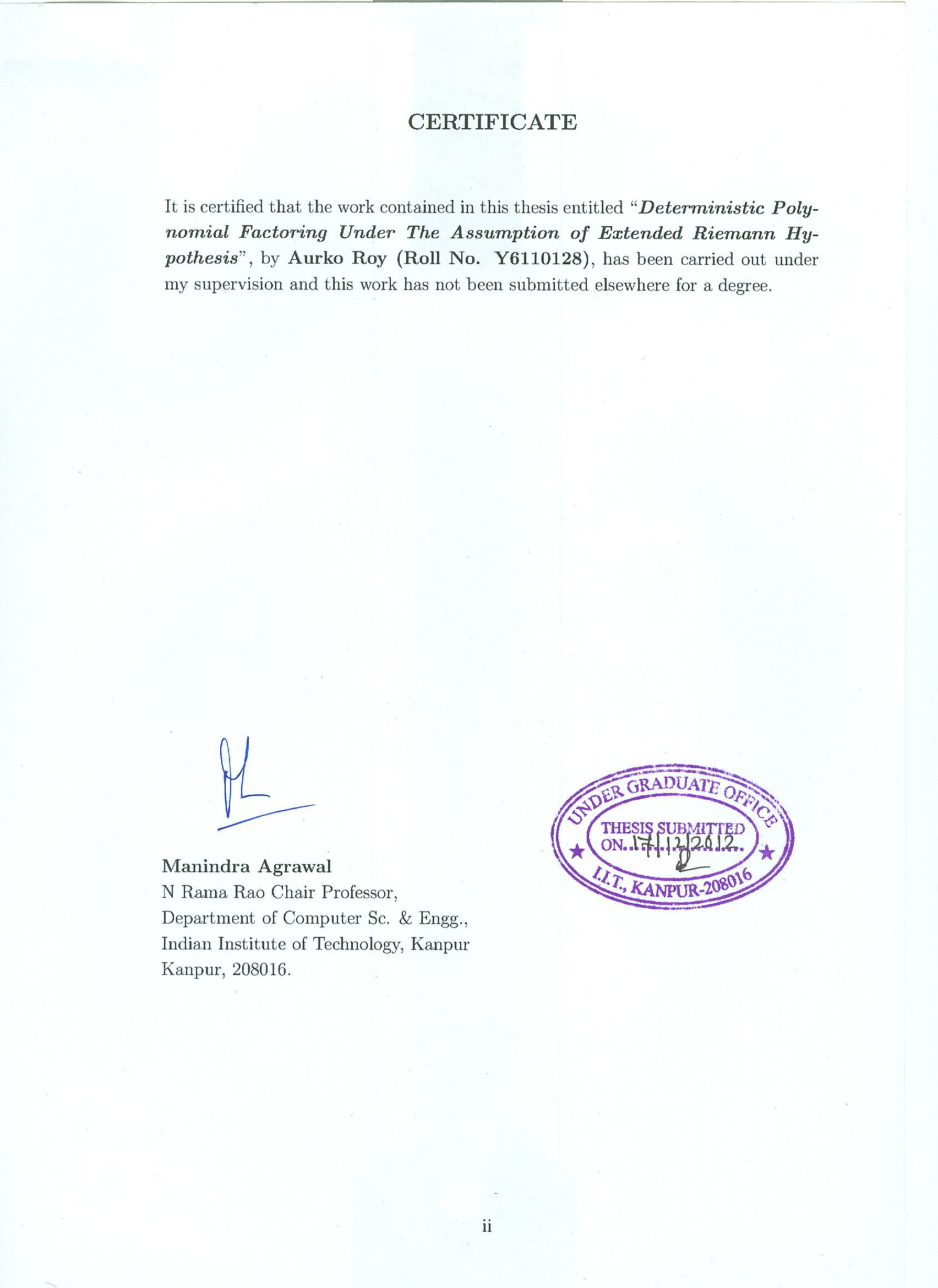}
\newpage
\addcontentsline{toc}{chapter}{Acknowledgements}
\begin{center}
\begin{large}
{\it{\bf ACKNOWLEDGEMENTS} }
\end{large}
\end{center} 
\textit{
	It has been a pleasure and an honor to be able to work with 
	Dr. Manindra Agrawal on this thesis and I am deeply indebted to him for his 
	guidance, support and his many valuable insights. I am also in the debt of 
	Ramprasad Saptharishi and Mrinalkanti Ghosh for numerous helpful discussions
	on several occasions. I am also grateful to my friends and family
	for their support and inspiration.
 }

\vskip 0.5in
\begin{flushright}
Aurko Roy\\
\end{flushright}

\newpage
\addcontentsline{toc}{chapter}{Abstract}
\chapter*{Abstract}
\begin{samepage} 
We consider the problem of deterministically factoring a univariate polynomial over a finite field under the assumption of the Extended Riemann Hypothesis (ERH). This work builds upon the line of approach first explored in
\cite{gao2001deterministic} and later expanded by \cite{saha2008factoring}. In both cases, the general approach has been to implicitly construct a graph with the roots as 
vertices and the edges formed by some polynomial time computable relation defined in \cite{gao2001deterministic}. Their algorithm then fails to factor a polynomial if this associated graph turns out to be \emph{regular}. In the 
first part of our work we strengthen the edge relation so that the resulting set of graphs we obtain are subgraphs of Gao's or Saha's graph, all of which must be \emph{regular}.

In the second part of our work we strengthen the regularity condition of these graphs from \emph{regular} to \emph{strongly regular}. This is accomplished by finding a parallel between their algorithms and the \emph{$1$-dimensional}
Weisfeiler-Leman algorithm for solving the Graph Isomorphism problem. We observe that the general principle behind their algorithms is to separate the roots by computing the \emph{$1$-dimensional} Weisfeiler-Leman approximation to the
orbits of this graph. This leads us to the natural question of whether this approximation may be improved. We then go on to show how to implicitly compute the \emph{$2$-dimensional} Weisfeiler-Leman approximation of the orbits of these
graphs.

The polynomials that this algorithm fails to factor form graphs that are \emph{strongly regular} and their set of adjacency matrices forms a combinatorial structure called an \emph{Association scheme}. This is closely related to the
structure of \emph{$m$-schemes} introduced by Saxena et al. in \cite{ivanyos2009schemes}, although our construction is more elementary than theirs. However, it should also be mentioned that their definition of 
\emph{$m$-schemes} is more general
than that of an \emph{association scheme}, although they primarily use structural results on \emph{association schemes} of prime order in their latest work \cite{arora2012deterministic}. In this sense our work serves as a sort of a 
bridge between the ideas of \cite{gao2001deterministic},\cite{saha2008factoring} and those presented in \cite{ivanyos2009schemes}, \cite{arora2012deterministic}.

We observe that polynomials whose schemes are \emph{thin} may be easily factored. In the final part
of our work we show that the problem of factoring a polynomial that gives rise to an arbitrary association scheme can always be polynomial time reduced to that of factoring another polynomial whose association scheme is 
\emph{primitive}. Primitive association schemes are simpler schemes that are generalizations of \emph{prime groups} and often (if the scheme is \emph{Schurian}) arise from \emph{primitive permutation groups}. Thus any extension
of the classification of prime schemes as presented in \cite{Hanaki:2006} to the more general case of primitive schemes would lead to better algorithms for factoring arbitrary polynomials.

\end{samepage}

\newpage
\addcontentsline{toc}{chapter}{Table of Contents}
\tableofcontents
\fancyhead[RO]{\thepage}
\fancyhead[LO]{\slshape \leftmark}
\fancyfoot[CO]{}
\renewcommand{\headrulewidth}{0.5pt}
\noindent
\pagenumbering{arabic}

\chapter{Introduction}

\section{The Problem}
The problem this thesis concerns itself with is to decompose a given polynomial over a finite field into its irreducible components. While this can always be accomplished in theory, from a computational perspective the 
interesting question is whether it can be done so in time \emph{polynomial} in its representation. Since a polynomial of degree $n$ over a field $\mathbb{F}_q$ can be represented optimally as a $n+1$ tuple of $\mathbb{F}_q$ elements,
we wish to factor it in time that is a polynomial function of $n$ and $\log{q}$. While this is an interesting problem in its own right, it finds several applications in computer algebra, algebraic coding theory, cryptography
and computational number theory. Polynomial factoring over finite fields appears in cyclic redundancy codes and BCH codes \cite{berlekamp1},\cite{macwilliams},\cite{vanlint}, in designing public key cryptosystems \cite{chor},
\cite{odlyzko},\cite{lenstra1} and in computing the number of points on elliptic curves \cite{buchmann}. More about the various applications of this problem may be found in \cite{von2001factoring}. It may be of some interest to note
that the same problem over $\mathbb{Q}$ has an efficient deterministic algorithm due to \cite{lenstra2}. We also note that this problem is of purely theoretical interest since for practical purposes various randomized
polynomial time algorithms exist due to \cite{ber70},\cite{cz81},\cite{gs92},\cite{ks95}.

\section{Previous Work}
The general problem of factoring a polynomial over an arbitrary finite field reduces deterministically in polynomial time to the problem of factoring a \emph{square-free} and \emph{completely splitting} polynomial over
a prime field due to \cite{ber70}. While efficient randomized algorithms exist, coming up with deterministic algorithms even under the assummption of the Extended Riemann Hypothesis (ERH) has met with limited success. In 
\cite{ronyai1988factoring}, R\'onyai proved under the assumption of ERH, that a polynomial with a bounded number of irreducible factors could be factored in polynomial time. More specifically he showed that for a polynomial $f \in \mathbb{F}_p$ of degree
$n$ and any prime divisor $m|n$, it was possible to obtain a factor of $f$ in time polynomial in $n^m$ and $\log{p}$. This in turn implied that \emph{even degree} polynomials could be easily factored. Under the assumption
of ERH, Huang in \cite{huang1985riemann} showed that it was possible to deterministically factor $n^{th}$ cyclotomic polynomials over $\mathbb{F}_p$ in polynomial time. R\'onyai \cite{ronyai1992galois} proved that under ERH
a polynomial with integer coefficients that generates a Galois number field may be factored \emph{mod} $p$ in deterministic polynomial time except for finitely many primes $p$ extending previous work by \cite{huang1991generalized},
\cite{adleman1977taking} and \cite{evdokimov1989factoring}. On special fields Bach, von zur Gathen and Lenstra \cite{bach2001factoring} proved that polynomials over finite fields of characteristic $p$ can be factored in polynomial
time if $\Phi_k(p)$ is \emph{smooth} for some integer $k$ where $\Phi_k(x)$ denotes the $k^{th}$ cyclotomic polynomial. This extended the work of von zur Gathen \cite{von1987factoring} who had shown how to factor polynomials 
\emph{mod} $p$ if $\Phi_1(p)=p-1$ is \emph{smooth}, which was also extended by Mignotte and Schnorr \cite{mignotte1988calcul} and R\'onyai \cite{ronyai1989factoring}. Evdokimov in \cite{evdokimov} gave a 
subexponential algorithm to deterministically factor polynomials under the assumption of ERH: a polynomial of degree $n$ over $\mathbb{F}_q$ could be factored in time polynomial in $n^{\log{n}}$ and $\log{q}$. Gao in 
\cite{gao2001deterministic} defined a class of polynomials whose roots satisfied a certain symmetry condition (which he named \emph{square balance}), such that any polynomial not belonging to this class could be factored in
polynomial time. This was generalized further by \cite{saha2008factoring} into the notion of \emph{cross balanced} polynomials in which this symmetry condition was strengthened. In \cite{ivanyos2009schemes},
\cite{arora2012deterministic} the problem of deterministic polynomial factoring was reduced to the study of certain combinatorial objects called 
\emph{$m$-schemes} whose properties were then exploited to factor polynomials of prime
degree $n$, where $n$ has a ``large'' \emph{$r$-smooth} divisor in time polynomial in $n^r$ and $\log{q}$.

\section{Our Contribution}
In this work we start off by generalizing the line of approach started by \cite{gao2001deterministic} and \cite{saha2008factoring} which ultimately however, lands us closer to the combinatorial structure of \cite{ivanyos2009schemes}.
Our approach is also motivated by the some of the approaches used to solve the \emph{Graph Isomorphism problem}, especially the Weisfeiler-Leman algorithm \cite{weisfeiler1968reduction} which we shall elucidate further in this
section. Further discussion on the graph isomorphism problem and the Weisfeiler-Leman algorithm may be found in \cite{babai1980isomorphism}.

The general approach in both \cite{gao2001deterministic} and \cite{saha2008factoring} is to implicitly construct a graph where the vertices are the roots (or certain polynomial time computable functions of them). Two vertices
are joined by an edge if they satisfy a certain relation: namely if $\xi_i,\xi_j$ be two roots (vertices) then there is an edge between them if $\sigma((\xi_i-\xi_j)^2)=\xi_i-\xi_j$ where $\sigma$ is a square root algorithm
defined in \cite{gao2001deterministic}. Then their algorithm fails if this graph (or set of graphs, considering polynomial functions of roots) turns out to be regular. Our first contribution is to strengthen this relation by 
giving a stronger test than Gao's square root algorithm (see section~\ref{sec:stronger}) which is simultaneously applicable to and improves both \cite{gao2001deterministic} and \cite{saha2008factoring}.

We then investigate how the \emph{regularity} condition may be improved. We observe that both the approaches in \cite{gao2001deterministic} and \cite{saha2008factoring} basically use an implicit form of the 
\emph{$1$-dimensional} 
Weisfeiler-Leman algorithm to enforce the \emph{square balance} and \emph{cross balance} symmetry conditions. This is not surprising since essentially we wish to separate out the roots and one way to separate them
would be to implicitly construct a graph with the roots as vertices and some polynomial time computable relation forming the edges, and thereafter implicitly trying to compute approximations to the orbits of this graph.

This line of thinking leads us to question whether it is possible to implicitly compute better approximations to the orbits and we demonstrate how to compute the \emph{$2$-dimensional} Weisfeiler-Leman approximation for this set of
implicit graphs (see section~\ref{sec:weisfeiler}). This approach then fails to factor a polynomial if the graphs of \cite{gao2001deterministic},\cite{saha2008factoring} turn out to be \emph{strongly regular}. Further,
the polynomials that resist factoring give rise to a combinatorial structure called \emph{Association schemes} (hence the similarity with \cite{ivanyos2009schemes}). Association schemes are generalizations of groups, in that the
set of all groups forms a subclass of association schemes (called \emph{thin schemes}). We then study some of the properties of the schemes that arise in this fashion, and prove that the 
problem of factoring a polynomial that gives rise to an arbitrary association scheme can always be polynomial time reduced to that of factoring a polynomial that gives rise to a \emph{primitive} association scheme. \emph{Primitive}
association schemes are a generalization of \emph{prime groups}, which we shall elucidate further in the subsequent sections (see subsections~\ref{sec:scheme} and \ref{sec:subset}). It is also of some interest to note (especially in
the context of \cite{arora2012deterministic}) that a
polynomial of prime degree $n$ which resists factoring by this method would always form a \emph{primitive} association scheme, although all primitive association schemes need not arise in this fashion. Hence according to this thesis
the study of factoring seems to be intricately related to the study of \emph{primitive association schemes}.

\section{Organization of Thesis}
In Chapter $2$ we will develop the preliminaries and the background necessary to understand the results presented in chapter $3$. We conclude our work in chapter $4$.

\chapter{Preliminaries}

\section{Extended Riemann Hypothesis}\label{sec:erh}

\begin{definition}
	The Riemann zeta function $\zeta$ is given by
	\begin{align*}
		\zeta(s)=\sum_{n=1}^\infty \frac{1}{n^s}=\prod_{p:prime} \frac{1}{1-p^{-s}}
	\end{align*}

	where $s \in \mathbb{C}$ and $Re(s) > 1$.
\end{definition}

Observe that if $Re(s) \le 1$ then the above formulation of $\zeta(s)$ is not convergent. However by \emph{analytic continuation} one may extend its definition to the entire complex plane. The $\zeta$ function also satisfies this
elegant functional equation (first discovered by Riemann in $1859$):

\begin{align*}
	\Gamma\left(\frac{s}{2}\right)\zeta(s)\pi^{-s/2} = \Gamma\left(\frac{1-s}{2}
	\right)\zeta(1-s)\pi^{-(1-s)/2}
\end{align*}

This immediately tells us of the existence of the zeroes of the $\zeta$ function at $s=-2,-4,-6,\cdots$ (i.e, all negative, even integers) since the $\Gamma$ function is singular for all non-positive integers. These are the
so called ``trivial zeroes''. The other zeroes $s=\sigma +i t$ of the $\zeta$ function lie in the \emph{critical strip} $0 \le \sigma \le 1$. The Riemann hypothesis talks about the distribution of these ``non-trivial zeroes''.

\vskip 2mm
\noindent \textbf{Riemann Hypothesis}. \textit{The non-trivial zeroes of $\zeta(s)$ have $Re(s)=\frac{1}{2}$.}

\vskip 3mm
The formulation of the Extended Riemann Hypothesis can be given in a number of ways. Here we will give one in terms of the \emph{Dirichlet} $L-function$. If $\chi$ is a character on $\mathbb{Z}/k\mathbb{Z}^*$ then it can be
extended to $\mathbb{Z}$ in the following fashion.

\begin{align*}
	\chi(m) = \begin{cases}\chi\left(m \mod{k}\right),  \quad 
	\text{if } \gcd(m,k)=1\\
	         0 \qquad \text{otherwise}
	         \end{cases}
	\end{align*}

\begin{definition} The Dirichlet $L$-function at level $k$ is defined as 

	\begin{align*}
		L(s,\chi) = \sum_{n \ge 1} \frac{\chi(n)}{n^s} = \prod_{p:prime} \frac{1}{1-\chi(p)p^{-s}}
	\end{align*}

	where $\chi$ is a character on $\mathbb{Z}/k\mathbb{Z}^*$ and the second equality follows from the multiplicative nature of the $\chi$.
\end{definition}

Similar to the $\zeta$ function, one may extend the definition of $L(s,\chi)$ by analytic continuation to the whole complex plane. It is then possible to formulate the Extended Riemann Hypothesis in a manner similar to the Riemann
Hypothesis for the ordinary $\zeta$ function.

\vskip 2mm
\noindent \textbf{Extended Riemann Hypothesis}. \textit{The non-trivial zeroes of $L(s,\chi)$ in the interval $0 \le Re(s) \le 1$ all have $Re(s)=\frac{1}{2}$.}

\vskip 3mm
It was proved by Ankeny in \cite{ankeny1952least} that if the Extended Riemann Hypothesis were true, then there exists a $q^{th}$ non-residue $a \in \mathbb{F}_p$ such that $a \in O(\log^2{p})$, where $q$ is a prime divisor of 
$p-1$. We will be using this result through out the rest of the work. 

\section{Semisimple Algebras}
We will encounter several examples of \emph{semisimple} algebras over fields as well as over rings, so we will define them below. To do that we will first define \emph{simple 
algebras} which are in some sense the building blocks of \emph{semisimple algebras}.

\begin{definition}\label{simplealg}
	Let $S$ be an algebra over a ring $R$. Then $S$ is a \textbf{simple algebra} if it has no proper two-sided ideals  and the set $S^2 \neq \{0\}$. 
\end{definition}

\begin{lemma}\label{simple}
	Let $R$ be a ring and let $S$ be a simple algebra over $R$ such that $S$ has a non-trivial center. Then $S=Se$ where $e\in S$ is the identity element in $S$, i.e. every simple algebra is unital.
\end{lemma}

\begin{proof} Since $S^2$ is an ideal of $S$ and $S^2 \neq \{0\}$ by definition \ref{simplealg}, we must have $S^2=S$. Consider the set $I = \{a | SaS=0\}$. $I$ is a two-sided ideal of $S$ and therefore $I=\{0\}$ or $I=S$. If
	$I=S$, then it would imply $S$ is nilpotent, hence $I=\{0\}$. Then for any $0 \neq a \in Z(S)$, $SaS=aS=Sa=S$.  
	Since $Sa=S$,  $\exists e_1 \in S$ such that $e_1.a=a$ and since $aS=S$, $e_1$ acts
	as the left identity of elements in $S$. Since $aS=S$, $\exists e_2 \in S$ such that $a.e_2=a$ and since $Sa=S$, $e_2$ acts as the right identity of elements of $S$. Finally $e_1.e_2=e_1=e_2$ so that $\exists e \in Z(S)$
	such that $\forall a \in S, e.a=a.e=a$.
\end{proof}

We will now define \emph{semisimple algebras}, which are algebras that are completely decomposable as a direct sum of \emph{simple algebras}.

\begin{definition}
	Let $R$ be a ring and $S$ an $R-algebra$. Then $S$ is called \textbf{semisimple} if $S \equiv \oplus_{i=1}^n S_i$ where each $S_i$ is a simple $R-algebra$.
 \end{definition}

 \begin{lemma}\label{idempotent} Let $S$ be a semisimple algebra over a ring $R$. Then there exists $\{E_1,\cdots,E_n\}\subset Z(S)$ such that $S=\oplus_{i=1}^n SE_i$ where $\forall 1 \le i \le n, E_i^2=E_i$ and
	 $E_iE_j=0$ if $i\neq j$.
 \end{lemma}

 \begin{proof}
	 Since $S$ is \emph{semisimple} we can decompose it as $S=\oplus_{i=1}^nS_i$ where each $S_i$ is a \emph{simple} $R-algebra$. By lemma \ref{simple} let the identity of each $S_i$ be $e_i$. Then define $E_i=(0,\cdots,0,e_i,
	 0,\cdots,0)$, where $e_i$ is in the $i^{th}$ position. The lemma then follows from this definition.
 \end{proof}
\vskip 2mm

One of the algebras we will keep encountering in the context of polynomial factoring is the $\mathbb{F}_p$ algebra $\mathcal{R}$ which we define below. Let $f=\prod_{i=1}^n(x-\xi_i)$ be the \emph{squarefree},
\emph{monic}, \emph{completely reducing} polynomial in $\mathbb{F}_p[x]$ that we wish to factor.

\begin{definition}\label{algebraR}
	 $\mathcal{R} \equiv \mathbb{F}_p[x]/(f(x)) \equiv \mathbb{F}_p[X]$ where $X \equiv x \pmod{f}$. 
\end{definition}

This algebra is also expressible as $\mathcal{R} \equiv \oplus_{i=1}^n\mathbb{F}_p[x]/(x-\xi_i)$, where for $1 \le i \le n$, $\mathbb{F}_p[x]/(x-\xi_i)$ is a \emph{simple} $\mathbb{F}_p-algebra$. This gives us the existence of 
the primitive idempotents of $\mathcal{R}$ over $\mathbb{F}_p$ which we will denote by the set $\{\mu_1,\cdots,\mu_n\}$, with the following properties $\forall 1 \le i,j \le n$

\begin{align*}
	\mu_i^2 &= \mu_i, \\
	\mu_i\mu_j &= 0, \quad i\neq j \\
	\sum_{i=1}^n \mu_i &= 1.
\end{align*}

It is also possible to work out the explicit expressions for these primitive idempotents.

\begin{align*}
	\mu_i = \prod_{j\neq i} \frac{x-\xi_j}{\xi_i-\xi_j}
\end{align*}
 
Where the denominator is invertible because $f$ is \emph{squarefree}. We will also talk about polynomials over $\mathbb{F}_p$ algebras which have a similar property, but before that we need to introduce the notion of 
\emph{separable} polynomials.

\begin{definition}
	Let $g$ be a monic univariate polynomial over an algebra $R$ with zeroes at $\chi_1,\cdots,\chi_n$. Then the discriminant $D(g)$ over $R$ is defined as $D(g) = \displaystyle\prod_{i,j,i\neq j} (\chi_i-\chi_j)$. The polynomial
	$g$ is called separable if $D(g)$ is a unit in $R$.
\end{definition}

\begin{lemma}
	Let $R$ be a semisimple $\mathbb{F}_p-algebra$. Let $g \in R[y]$ be a completely splitting, separable, monic polynomial. Then the $R-algebra$ $S \equiv R[y]/(g)$ is also semisimple.
\end{lemma}

\begin{proof} Let $g=\prod_{i=1}^n (y-\chi_i)$ where $\chi_i \in R$. Note that this decomposition need not be unique. Define $e_i$ as follows:

	\begin{align*}
		e_i = \prod_{j\neq i} \frac{y-\chi_j}{\chi_i-\chi_j}.
	\end{align*}

	The $e_i$'s are well defined since $g$ is separable. Then the $R-algebra$ $S$ 
	may be written as $S \equiv \displaystyle \oplus_{i=1}^n Se_i$ where each 
	$Se_i$ is a \emph{simple} $R-algebra$.

\end{proof}

\section{GCD of Polynomials over Algebras}
In this section we will talk about what we mean by taking the $\gcd$ of two polynomials over \emph{semisimple algebras} over $\mathbb{F}_p$. Note that the notion of $\gcd$ need not even make sense for polynomials over arbitrary 
algebras, but it is possible to extend the notion of $\gcd$ meaningfully over \emph{semisimple algebras}. Let $R$ be a semisimple algebra over $\mathbb{F}_p$. Then by lemma \ref{idempotent} there exist a set of orthogonal primitive
idempotents $\{e_1,\cdots,e_n\}$ such that $R = \oplus_{i=1}^n Re_i$. Let $g,h$ be polynomials over $R[y]$. Then we can write them as 

\begin{align*}
	g &=\sum_{i=1}^n g_i e_i \\
    h &= \sum_{i=1}^n h_i e_i
\end{align*}

Where $\forall 1 \le i \le n$, $g_i,h_i \in \mathbb{F}_p[y]$. This allows us to generalize the notion of the $\gcd$ in the most obvious way.

\begin{definition}\label{gcd}
	Let $R$ be an $\mathbb{F}_p$ algebra and $g,h$ as above polynomials in $R[y]$. Then by the $\gcd(g,h)$ in $R[y]$ we mean
	
\begin{align*}
	\gcd(g,h) = \sum_{i=1}^n \gcd(g_i,h_i)e_i
\end{align*}

where $\gcd(g_i,h_i)$ is the usual $\gcd$ over the ring $\mathbb{F}_p[y]$.
\end{definition}

Using definition \ref{gcd} we can now talk about the $\gcd$ of polynomials over the algebra $\mathcal{R}$ (definition \ref{algebraR}) as well as over \emph{semisimple algebras} over $\mathcal{R}$. However note that we do not
explicitly know the idempotents $\{\mu_1,\cdots,\mu_n\}$ of $\mathcal{R}$ over $\mathbb{F}_p$ so we cannot directly compute the expression in definition \ref{gcd}.

\begin{lemma}\label{gcdbasis}
	Let $\mathcal{S}$ be any semisimple algebra over $\mathcal{R}$ of dimension polynomial in $n=deg(f)$ for which we have a basis over $\mathbb{F}_p$. Let $g,h \in \mathcal{S}[y]$ be any two polynomials such that 
	$deg(g),deg(h)$ is bounded by some polynomial function of $n$. Then their $\gcd$ as defined in \ref{gcd} may be calculated in time polynomial in $n$ and $\log{p}$.
\end{lemma}

\begin{proof}
	Run the usual Euclidean algorithm on the two polynomials $g,h$. The procedure either goes through or we encounter a polynomial whose leading coefficient $a \in \mathcal{S}$ is not invertible. Consider the two orthogonal 
	algebras $\mathcal{S}_1=\mathcal{S}a$ and $\mathcal{S}_2=\mathcal{S}-\mathcal{S}a$, so that $\mathcal{S}= \mathcal{S}_1 \oplus \mathcal{S}_2$. Since we have a basis for $\mathcal{S}$ over $\mathbb{F}_p$ it is easy to compute
	the identity elements $e_1$ and $e_2$ of $\mathcal{S}_1$ and $\mathcal{S}_2$ respectively. Let $g_1=ge_1,h_1=he_1$ and $g_2=ge_2,h_2=he_2$. Recursively compute $\gcd(g_1,h_1) \in \mathcal{S}_1[y]$ and $\gcd(g_2,h_2) \in 
	\mathcal{S}_2[y]$. Then it is easy to see from the definition that $\gcd(g,h) = \gcd(g_1,h_1) + \gcd(g_2,h_2) \in \mathcal{S}[y]$ and that it takes time bounded by some polynomial in $n$ and $\log{p}$.
\end{proof}
	
Note that although we do not know the idempotent basis of $\mathcal{R}$ over $\mathbb{F}_p$ we do have the usual polynomial basis $\{1,x,x^2,\cdots,x^{n-1}\}$ over $\mathbb{F}_p$. Hence we may use the procedure in lemma \ref{gcdbasis}
to calculate $\gcd$'s of polynomials over $\mathcal{R}$.

\chapter{Main Results}

Let $f=\prod_{i=1}^n (x-\xi_i) \in \mathbb{F}_p[x]$ be the \emph{square-free}, \emph{monic} and \emph{completely splitting} polynomial which we wish to factor. Before we deal with the notion of \emph{square balanced polynomials} and
their extension we will make a comment about $f$.

\begin{lemma}\label{polyroot}
	Let $q \in \mathbb{F}_p[x]$ be a polynomial of degree bounded by some polynomial in $n$ and $\log{p}$. Then the polynomial $f_q(x) = \prod_{i=1}^n (x-q(\xi_i)) \in \mathbb{F}_p[x]$ may be constructed in time polynomial in $n$
	and $\log{p}$.
\end{lemma}

\begin{proof}
	Given $f$ we can construct its companion matrix $C_f\in Mat_n(\mathbb{F}_p)$. Then $f_q(x) = det(xI-q(C_f))$ by definition. Since $deg(q)$ is bounded by some polynomial in $n$ and $\log{p}$ the whole
	operation can be done in time polynomial in $n$ and $\log{p}$.
\end{proof}

\begin{lemma}\label{polyrootfactor}
	Let $f_q \in \mathbb{F}_p[x]$ be the polynomial as defined above. If we can obtain $g_q$, a non-trivial factor of $f_q$, then we can obtain a non-trivial factor of $f$ in additional time that is polynomial in $n$ and $\log{p}$.
\end{lemma}

\begin{proof}
	It is easy to see that $\gcd(g_q(q(x)),f)$ gives us the required non-trivial factor.
\end{proof}

From these two simple lemmas we conclude that factoring $f$ is equivalent to factoring a polynomial $f_q$ where the degree of $q$ is bounded by a polynomial in $n$ and $\log{p}$. These polynomials form the basis of Gao's
\emph{super square balance} condition in \cite{gao2001deterministic} and Saha's \emph{cross balance condition} in \cite{saha2008factoring}, since the symmetry condition on $f$ can always be extended to a similar symmetry condition
on $f_q$ (for suitable polynomial $q$) due to lemmas \ref{polyroot} and \ref{polyrootfactor}. Since we are interested in strengthening this symmetry condition, we will henceforth not talk about these polynomials but will note that
all the following results for $f$ are also applicable to polynomials of the form $f_q$.

\section{A Stronger Notion of Balanced Polynomials}\label{sec:stronger}
Let $\mathcal{R}$ be the $\mathbb{F}_p-algebra$ defined in definition \ref{algebraR}. Let
$p-1=2^rw$, where $w$ is odd and let $\gamma$ be a \emph{quadratic non-residue} that can be computed effeciently due to ERH (see section~\ref{sec:erh}). Then $\eta=\gamma^w$ is a generator of the \emph{$2$-Sylow} group $\mathbb{F}_p^\times$.
In \cite{gao2001deterministic} Gao gave an algorithm $\sigma$ to compute the square roots of elements in the alegbra $\mathcal{R}$, where $\sigma(a)$ is the square root given by the algorithm for any quadratic residue $a \in
\mathcal{R}$. Then we have the following lemma due to Gao.

\begin{lemma}\label{squareroot}
	Let $a \in \mathbb{F}_p$ such that $a=\eta^u\theta$ where $\theta$ has odd order. Let $u = \sum_{i=0}^{r-1} u_i2^i$, where $\forall 1 \le i \le r-1, u_i \in \{0,1\}$. Then $\sigma(a^2)=a$ iff $u_{r-1}=0$ and 
	$\sigma(a^2)=-a$ iff $u_{r-1}=1$.
\end{lemma}

\begin{proof}The proof can be found in \cite{gao2001deterministic}. We will omit the proof here since we will not be using his algorithm.
\end{proof}

Based upon this algorithm Gao gave the following definition for \emph{square balanced} polynomials.

\begin{definition}\label{squarebalance}
	Let $f=\prod_{i=1}^n (x-\xi_i) \in \mathbb{F}_p[x]$ be a \emph{square-free},\emph{completely splitting} polynomial. Then for each $1 \le i \le n$ define the set $D_i$ as follows:

	\begin{align*}
		D_i &= \left\{\xi_j \mid \xi_i\neq \xi_j,
		\sigma( (\xi_i-\xi_j)^2)=\xi_i-\xi_j\right\}
	\end{align*}

	Then the polynomial $f$ is square-balanced if $\forall 1 \le i,j \le n, |D_i|=|D_j|=\frac{n-1}{2}$.
\end{definition}

\noindent
Lemma \ref{squareroot} also gives us this alternative definition of the sets 
$D_i$:

\begin{align*}
	D_i &= \left\{\xi_j |\xi_i \neq \xi_j,(\xi_i-\xi_j)= \eta^u \theta,2\nmid o(\theta),u=\sum_{k=0}^{r-1}u_k2^k,u_{r-1}=0\right\}.
\end{align*}

We will now generalize this definition of the sets $D_i$ given by Gao and give a ``balance condition'' that is stronger.

\begin{definition}\label{strongerbal}
	Consider the following sequence of sets for $0 \le k \le r-1$
\begin{align*}
	D^k_i &= \left\{\xi_j \mid\xi_i \neq \xi_j,(\xi_i-\xi_j)=\eta^u 
	\theta,2\nmid o(\theta),u=\sum_{k=0}^{r-1}u_k2^k,u_{k}=0\right\}.\\
\end{align*}

Then the polynomial $f$ satisfies this stronger ``balance condition'' if $\forall 1 \le i,j \le n, \forall 0 \le k \le r-1, |D^k_i|=|D^k_j|$.
\end{definition}

Note that Gao's \emph{square balance} condition is essentially $\forall 1 \le i,j \le n,|D^{r-1}_i|=|D^{r-1}_j|$ and hence is a special case of definition \ref{strongerbal}. The following lemma is an extension of Gao's 
result on \emph{square balanced} polynomials.

\begin{lemma}\label{g}
	A polynomial $f \in \mathbb{F}_p[x]$ can be factored in deterministic polynomial time under the assumption of ERH, if it does not satisfy the``balance condition'' of definition \ref{strongerbal}.
   \end{lemma}

\begin{proof}
	Let $X \equiv x \pmod{f}$. Consider $g(y,x)=\frac{f(-y+X)}{-y} \in \mathcal{R}[y]$. We have

	\begin{align*}
		g(y,x) &= \frac{f(-y+X)}{-y} = \sum_{i=1}^n \prod_{j\neq i} (y-(\xi_i-\xi_j)) \mu_i
	\end{align*}

	We define a sequence of polynomials $d_0,\cdots,d_k,\cdots,d_{r-1}$ in the following fashion:

	\begin{align*}
	       d_k &= \gcd\left(\prod_{i=0}^{2^k-1} \left(y^{\frac{p-1}{2^{k+1}}}-
	       \eta^i\right),g\right) = 
	       \sum_{i=1}\prod_{j\in D^k_i}(y-(\xi_i-\xi_j))\mu_i
	     \end{align*}

	From definition \ref{strongerbal} it is not too difficult to see that each polynomial $d_k$, $0 \le k \le r-1$ corresponds to the sequence of sets $D^k_i$, $1 \le i \le n$ since $d_k(-y+X)\mu_i$ has its roots the
	set $D^k_i$. 
	     Since (by the assumption of ERH) we know $\eta$, each of these 
	     \(\gcd\)'s may be calculated in time polynomial in $n$ and
	     $\log{p}$ and the number of such polynomials ($2r$ of them) is bounded by $O(\log{p})$. Now suppose if $\exists 0 \le k \le r-1,\exists 1\le i,j \le n$ such that  $|D^k_i|\neq |D^k_j|$ then the leading coefficient
	     of $d_k$ is a \emph{zero divisor} in $\mathcal{R}$ which gives a decomposition of $f$ over $\mathbb{F}_p$.
\end{proof}

The above lemma leads to algorithm \ref{alg:squarebal} which fails to factor a polynomial $f \in \mathbb{F}_p[x]$ if it satisfies the above stronger notion of symmetry or balance.
A slight modification of the same approach also allows us to infer that the roots of $f$ must have the same \emph{$2$-Sylow} component.

\begin{lemma}
	A polynomial $f = \prod_{i=1}^n (x-\xi_i) \in \mathbb{F}_p[x]$ can be factored in deterministic polynomial time under the assumption of ERH if $\exists i,j \in \{1,\cdots,n\},\xi_i^w \neq \xi_j^w$, or in other
	words if the \emph{$2$-Sylow} component of $\xi_i$ differs from that of $\xi_j$.
\end{lemma}

\begin{proof}
	Consider a sequence of polynomials $s_0,s_1,\cdots,s_k,\cdots,s_{r-1} \in \mathbb{F}_p[x]$ given by:

	\begin{align*}
		s_k &= \gcd\left(\prod_{i=0}^{2^k-1} \left(x^{\frac{p-1}{2^{k+1}}}-\eta^i
		\right), f \right) 
	\end{align*}

	The only way we do not get a factor of $f$ in this fashion is if $\forall 1 \le k \le r-1$, $s_k$ is either
	$1$ or $f$. Clearly this implies that the \emph{$2$-Sylow} expansion of every root is the same.
\end{proof}

The notion of lemma \ref{g} leads us to the following algorithm which fails to factor a polynomial $f$ if it satisfies the symmetry condition of definition \ref{strongerbal}.

\begin{algorithm}[H]
	\caption{The Stronger Square Balance Algorithm}
	\label{alg:squarebal}
\begin{algorithmic}
         \State  k $\leftarrow$ $0$ 
	 \State  $g_0 \leftarrow \gcd(y^{\frac{p-1}{2}}-1,g)$
	 \State  $g_1 \leftarrow \gcd(y^{\frac{p-1}{2}}+1,g)$
 \If {$g_0\neq 1$} \State $(0)\in S_0$ \EndIf
 \If {$g_1\neq 1$} \State $(1)\in S_0$ \EndIf
 \For {$1 \le k \le r-1$}
	 \For {$(u_0,\cdots,u_{k-1})\in S^{k-1}$}
		  \State $g_0 \leftarrow \gcd(y^{\frac{p-1}{2^{k+1}}}-\eta^{\sum_{j=0}^{k-1}u_{j}2^{r-(k-j+1)}},g_{k-1})$
		  \State $g_1 \leftarrow \gcd(y^{\frac{p-1}{2^{k+1}}}-\eta^{\sum_{j=0}^{k-1}u_{j}2^{r-(k-j+1)}+2^{r-1}},g_{k-1})$
		  \If {$g_0 \neq 1$} 
		      \If {$k=r-1$} 
			\State $g_0 \in \mathcal{E}$ 
		      \EndIf
		      \State $(u_0,\cdots,u_{k-1},0) \in S_k$ 
		  \EndIf
		  \If {$g_1 \neq 1$} 
		      \If {$k=r-1$} 
		          \State $g_1 \in \mathcal{E}$
		      \EndIf
		      \State $(u_0,\cdots,u_{k-1},1) \in S_k$ 
		  \EndIf
	 \EndFor
\EndFor
\end{algorithmic}
\end{algorithm}

\vskip 1cm
This algorithm terminates in time polynomial in $n$ and $\log{p}$ since $|\mathcal{E}|=|S_{r-1}| \le n-1$ and so atmost $n-1$ of the branches may be explored. The length of each branch is $r \le \log{p}$ and
each of the gcd's are calculable in time polynomial in $n$ and $\log{p}$. Consider the set $\mathcal{E}$. Since $\xi_i-\xi_j$ and $\xi_j-\xi_i$ differ in atleast one bit of their \emph{$2$-Sylow} expansion we have that
$2 \le |\mathcal{E}| \le n-1$. Consider a polynomial $g_l \in \mathcal{E}$.

\begin{align*}
	g_l = \sum_{i=1}^n \prod_{j} (y-(\xi_i-\xi_j)) \mu_i
\end{align*}

These polynomials give rise to a sequence of $n \times n$ matrices $\{E_l\}_{l=1}^{|\mathcal{E}|}$ defined below.

\begin{definition}\label{matrixdef}
	$E_l(i,j)=1$ if $\xi_i-\xi_j$ is a $\mathbb{F}_p$ root of the polynomial $g_l\mu_i \in \mathcal{R}[y]$; otherwise $E_l(i,j)=0$. In other words $E_l(i,j)=1$ iff  $g_l(\xi_i-\xi_j)\mu_i=0$ and $E_l(i,j)=0$ otherwise.
\end{definition}

Since the polynomial and matrix representation are equivalent, we will talk of the set $\mathcal{E}$ interchangeably as consisting of the polynomials $g_l$ or the matrices $E_l$.We then have the following two lemmas regarding some
properties of the set $\mathcal{E}$.

\begin{lemma} \label{welldef}
	$\forall 1\le i,j,s,t \le n$ if $E_l(i,j)=E_l(s,t)=1$ then $(\xi_i-\xi_j)^w=(\xi_s-\xi_t)^w$. Conversely, if $E_l(i,j)=E_m(s,t)=1$ and
	$(\xi_i-\xi_j)^w=(\xi_s-\xi_t)^w$ then $E_l=E_m$.
\end{lemma}

\begin{proof}
	The forward direction follows since if $g_l(\xi_i-\xi_j)=g_l(\xi_s-\xi_t)=0$ then by algorithm \ref{alg:squarebal} we conclude that $\xi_i-\xi_j$ and $\xi_s-\xi_t$ must have the same \emph{$2$-Sylow} component. To see the
	converse note that if $\xi_i-\xi_j$ and $\xi_s-\xi_t$ do have the same 
	\emph{$2$-Sylow} component, then above algorithm fails to separate them and hence they would belong to the same polynomial in $\mathcal{E}$.
\end{proof}

\begin{lemma}\label{transpose}
	If $E_l \in \mathcal{E}$ then $E_l^\top \in \mathcal{E}$.
\end{lemma}

\begin{proof}
	Let  $E_l(i,j)=1$ for some $1 \le i,j \le n$. Then there exists some matrix $E_m$ such that $E_m(j,i)=1$. If $E_l$ or $E_m$ do not contain any other non-zero entries besides this then $E_m=E_l^\top$ and we are done.
	Otherwise let $E_l(s,t)=1$ for some $1 \le s,t \le n,s\neq i,t\neq j$. Then we know that $(\xi_i-\xi_j)^w=(\xi_s-\xi_t)^w$ or equivalently $(\xi_j-\xi_i)^w=(\xi_t-\xi_s)^w$ so that from lemma \ref{welldef} we conclude that 
	$E_m(s,t)=1$, so that $E_m=E_l^\top$ and we are done.
\end{proof}

It is also possible to visualize this in a graph-theoretic setting where the vertices
are the set of roots of $f \in \mathbb{F}_p[x]$ and the set $\mathcal{E}$ form a disjoint multiset of edges.

\begin{definition}\label{multigraph}
	Let $V$ a set of size $n$ labelled by the integers from $1$ to $n$. Consider $\mathcal{E}$ to consist of the matrices $E_l$ defined in \ref{matrixdef}. Then $G_f=(V,\mathcal{E})$ is a multigraph on $n$ vertices.
\end{definition}

\begin{lemma}
      If $f$ fails to be factored by algorithm \ref{alg:squarebal} then for every $E_l \in \mathcal{E}$, the restriction of $G_f$ to $(V,E_l)$ must be regular.
\end{lemma}

\begin{proof}
	This follows from algorithm \ref{alg:squarebal} and lemma \ref{strongerbal}.
\end{proof}

In this graph theoretic setting one can view algorithm \ref{alg:squarebal} in the following fashion. Given this set of graphs we wish to distinguish between and separate the vertices (the roots) in some fashion in order to obtain
a non-trivial
factor of $f$. One possible way would be to look at the orbits of every vertex and compare them by size. Algorithm \ref{alg:squarebal} simply computes out the $1$-dimensional \emph{Weisfeiler-Leman} approximation for the orbit 
of every vertex for each of the graphs $(V,E_l), 1 \le l \le |\mathcal{E}|$; the set of $E_l$ colored neighbors of every vertex being the $1$-dimensional approximation of its orbits. If the graph $E_l$ is not \emph{regular} then
one can separate the roots or equivalently obtain a non-trivial factor of $f$. A natural next step in generalizing this 
algorithm would be to come up with a better approximation for the set of orbits for each vertex and thereby tighten the symmetry condition under which factoring fails. 

\section{Weisfeiler-Leman and Factoring}\label{sec:weisfeiler}

In this section we build on the idea of the preceding section to come up with a better approximation for the size of orbits of each vertex by showing that it is possible to implicitly compute the $2$-dimensional 
\emph{Weisfeiler-Leman} approximation of the orbits. We briefly touch upon the general $2$-dimensional \emph{Weisfeiler-Leman} algorithm before discussing it in the context of polynomial factoring.

\subsection{2-dimensional Weisfeiler Leman}
A more thorough treatment of the general algorithm can be found here \cite{barbados}. Consider a multigraph $G=(V,\mathcal{E})$ ( where $\mathcal{E}$ is a set of colors) which satisfies the following properties:

  \begin{enumerate}
	  \item For every $E_i,E_j \in \mathcal{E},i\neq j, E_i \cap E_j = 
	  \emptyset$, i.e. the colors are \emph{disjoint}.

	  \item $\displaystyle\sum_{i:E_i\in \mathcal{E}} E_i = J$ where $J$ is the all $1$'s matrix.

  \end{enumerate}

  We wish to further refine this set into a set of colors $\mathcal{S}$ which satisfy the following \emph{well-behaved} property:

 \begin{enumerate}
	 \item The entries of each $S_i \in \mathcal{S}$ come from the set $\{0,1\}$.

	 \item $\forall S_i,S_j \in \mathcal{S}, i \neq j, S_i \cap S_j = \emptyset$
	  and $\displaystyle\sum_{i:S_i\in \mathcal{S}}^k S_i =J$, $J$ being the all $1$'s matrix.

	 \item Let $P$ be any automorphism of the multigraph $(V,\mathcal{E})$, i.e. $\forall E_i \in \mathcal{E}, PE_iP^{-1}=E_i$. Then the colors should be unchanged by $P$, i.e. $PS_iP^{-1}=S_i,\forall S_i \in \mathcal{S}$.
		 
 \end{enumerate}

 It is clear from the definition that the initial multigraph $G$ satisfies this \emph{well-behaved} property. This gives us the set of \emph{well-behaved} colors for the first iteration. Then the algorithm proceeds in the following
 fashion:

 \begin{algorithm}[H]
	 \caption{The 2-Dimenstional Weisfeiler Leman Algorithm}\label{wl}
 \begin{algorithmic}
     \State $\mathcal{C} \leftarrow \mathcal{E}$
     \While{True} 
     \State $k \leftarrow |\mathcal{C}|$
      \If {$\forall C_i,C_j \in \mathcal{C}$, $C_iC_j= \sum_{l=1}^k lC_{\alpha_l}$ for some set $\{\alpha_l\}$ of color indices} 
     \State Output $\mathcal{C}$ as the final set of colors
	 \Else 
	         \State $C_iC_j=\sum_{l=1}^k lD_l$ where each $D_l$ is a $\{0,1\}$ matrix.
		 \State $\mathcal{C} \leftarrow \mathcal{C} \cup_{l=1}^k \{D_l\}$. Let $\mathcal{C} = \{C'_1,\cdots,C'_{k'}\}$.
		 \While {$\exists 1\le i,j \le k',i \neq j, C'_i \cap C'_j \neq 
		 \emptyset$}
		 \State $\mathcal{C} \leftarrow (\mathcal{C}\setminus \{C'_i,C'_j\}) \cup \{D'_1=C'_i\setminus C'_j,D'_2=C'_j\setminus C'_i,D'_3=C'_i \cap C'_j\}$
		 \EndWhile
	\EndIf
     \EndWhile
 \end{algorithmic}
\end{algorithm}

This algorithm terminates in time polynomial in $n$, since the set of colors either increases by atleast one or the algorithm terminates and the size of the set of colors is bounded by $n$.
Let $\mathcal{C}$ be the final set of colors so obtained. It is clear that this set satisfies condition $1$ and $2$ of the \emph{well-behaved} properties. The next lemma shows that it is \emph{well-behaved}, 
 i.e. it satisfies all the properties.

 \begin{lemma}\label{wellbehaved}
	 The set of colors $\mathcal{C}$ so obtained by this algorithm is well behaved.
\end{lemma}

\begin{proof}
	Proof is by induction on the iteration steps. Since the set $\mathcal{E}$ is \emph{well behaved} the base assertion certainly holds. Suppose the claim holds for the set of colors $\mathcal{S}$ at the $i^{th}$ step. 
	Then according to the inductive hypothesis $\forall P\in Aut(G),\forall S_i \in \mathcal{S}, PS_iP^{-1}=S_i$, so that $\forall S_i,S_j \in \mathcal{S}, PS_iS_jP^{-1}=PS_iP^{-1}PS_jP^{-1}=S_iS_j$. However we have
	$S_iS_j = \sum_{l=1}^k lD_l$ so that $\sum_{l=1}^k lPD_lP^{-1}=\sum_{l=1}^klD_l$. Since each $D_l$ is a $0/1$ matrix, by identifying the matrix whose entries are $l$ on both sides we conclude that $PD_lP^{-1}=D_l$.
	Upon adding these new matrices let the set be $\mathcal{S}'$. If $\exists S'_i,S'_j \in \mathcal{S}',S'_i \cap S'_j\neq \emptyset$ then we replace $S'_i$
	 and $S'_j$ by $D'_1=S'_i\setminus S'_j$, $D'_2=S'_j\setminus S'_i$ and 
	$D'_3=S'_i\cap S'_j$. Consider an edge $(u,v) \in D'_3$ and any automorphism $P \in Aut(G)$. Then since $(u,v) \in S'_i,S'_j$ and $P$ preserves them both we have $P$ preserves $D'_3$. It therefore follows that $P$
	preserves $D'_1=S'_i\setminus D'_3$ and $D'_2=S'_j\setminus D'_3$.
\end{proof}

There is another property besides the \emph{well-behavedness} of this set that we need to show. We call a set of colors $\mathcal{S}$ closed under taking transposes, iff $S_i \in \mathcal{S} \Rightarrow S_i^\top \in \mathcal{S}$.
Note that the set $\mathcal{E}$ of the graph $G_f$ is closed under taking transposes due to lemma \ref{transpose}. We then have the following lemma.

\begin{lemma}\label{transposepol}
	If the original set of colors $\mathcal{E}$ is closed under taking transposes, then the set $\mathcal{S}$ obtained after every iteration of algorithm \ref{wl} remains closed under transposes.
\end{lemma}

\begin{proof}
	The proof is by induction on the iterations of the algorithm. The base case is true by assumption. Suppose at the $t^{th}$ iteration this assertion holds. Let the set of colors at that stage be 
	$\mathcal{S}$. Consider the additional colors $D_l$ that are introduced by multiplying $S_iS_j$ where $S_i,S_j \in \mathcal{S}$. Suppose 
	$S_iS_j=\sum_{l}lD_l$ then we also have $S_i^\top S_j^\top=\sum_{l}lD_l^\top$ so that if $D_l$ is added to the set, then so is $D_l^\top$. Let this new set be $\mathcal{T}$. Then this set by the previous
	argument is closed under taking transpose. Suppose that $\exists T_i,T_j \in \mathcal{T},T_i \cap T_j \neq \emptyset$, so that we replace $T_i,T_j$ with $T_i \cap T_j$, $T_i \setminus T_j$ and $T_j \setminus T_i$.
	This does not affect the transpose property
	since $T_i^\top \cap T_j^\top= (T_i\cap T_j)^\top$, $T_i^\top \setminus T_j^\top = (T_i \setminus T_j)^\top$ and $T_j^\top \setminus T_i^\top = (T_j \setminus T_i)^\top$. Therefore the set at the $t+1^{th}$ step retains
	its closure property under taking transposes.
\end{proof}

\subsection{Weisfeiler-Leman and Polynomials}
	In this section we illustrate how the steps of the $2$-dimensional \emph{Weisfeiler-Leman} may be carried out with polynomials. Consider the set $\mathcal{E}$, the set of polynomials $\in \mathcal{R}[y]$ obtained
	from algorithm \ref{wl}. Let $g_l \in \mathcal{E}$ be such a polynomial, then we have associated with it the matrix $E_l$ which was defined in definition \ref{matrixdef}.
        In the preceding section we defined the multigraph $G_f=(V,\mathcal{E})$. We then have the following lemma.

	\begin{lemma}
		The set $\{I\} \cup \mathcal{E}$ is well behaved.
	\end{lemma}

	\begin{proof}
		It is easily seen that every element of $\mathcal{E}$ when thought of as a matrix has its entries from $\{0,1\}$ so that condition $1$ of \emph{well behavedness} is satisfied. Further, we claim that for 
		any $E_l,E_m \in \mathcal{E},l\neq m, E_l \cap E_m = \emptyset$. Suppose if $E_l$ is identity then this is clearly true since $g_m$ is a factor of $g \in \mathcal{R}[y]$ and we know that $y \nmid g$ so that
		$\forall i,(\xi_i,\xi_i) \notin
		E_m$. If neither of them are the
		identity matrix then from lemma \ref{welldef} we have that $E_l=E_m$. Further $\forall 1 \le i,j \le n,i\neq j,\exists 1 \le l \le |\mathcal{E}|, E_l(i,j)=1$, since $\prod_{l=1}^{|\mathcal{E}|}g_l=g$
		. Hence it follows that $I+\sum_{E_l \in \mathcal{E}}E_l=J$ so that condition $2$ is satisfied as well. The last condition follows from the defintion of $G_f$.
	\end{proof}

	Hence one can think of applying the $2-dimensional$ \emph{Weisfeiler-Leman} algorithm in this case. Note however that we do not explicitly know the matrices $E_l$ but rather have their polynomial forms $g_l \in \mathcal{R}[y]$.
	Therefore we must show that the algorithm \ref{wl} may be duplicated with just this set of polynomials.

	\subsubsection{Multiplication of Matrices}
	Consider a \emph{well behaved} set of colors $\mathcal{S}=\{S_1,S_2,\cdots,S_m\}$ with respect to the complete graph on $n$ vertices. Assume also that this set is closed under taking transposes.	
	Suppose they are given in their polynomial forms, i.e. the polynomial corresponding to $S_l$ is 
	defined by 

	\begin{align*}
		g_l(y,x) &= \sum_{i=1}^n \prod_{j:S_l(i,j)=1} (y-(\xi_i-\xi_j)) \mu_i
	\end{align*}

        where $\mu_i$ depends on $x$. We now wish to show that from their polynomial forms we can recover the polynomial form of the matrix corresponding to $S_lS_t$ for some $S_l$ and $S_t$ in this \emph{well behaved} set. 
	Denote by $\bar{g}_l$ the 
	polynomial corresponding to $S_l^\top$. Consider the algebra $\mathcal{T} \equiv \mathcal{R}[y]/(\bar{g}_l) \equiv \mathcal{R}[Y]$, where $Y \equiv y \pmod{\bar{g}_l}$. Since $\bar{g}_l$ is square-free and completely
	splitting over $\mathcal{R}$, we have that $\mathcal{T}$ is a semisimple $\mathcal{R}$-algebra. Let its primitive idempotents over $\mathcal{R}$ be $\nu_1,\cdots,\nu_{d_l}$, where $deg(g_l)=deg(\bar{g}_l)=d_l$. Suppose
	$\bar{g}_l$ splits as below over $\mathcal{R}[y]$:
	
	\begin{align*}
		\bar{g}_l(y,x) &= \prod_{j'=1}^{d_l} (y-\chi_{j'}) = \sum_{i=1}^n \prod_{j:S_l^\top(i,j)=1} (y-(\xi_i-\xi_j))\mu_i
	\end{align*}

	where each of the $\chi_{j'} \in \mathcal{R}$ and hence can be written as $\chi_{j'}=\sum_{i=1}^n \chi_{ij'} \mu_i$. It then follows that $Y = \sum_{j'=1}^{d_l} \chi_{j'}\nu_{j'}$. It also follows from the expression 
	for $\bar{g}_l$ above that the set $\{\chi_{i1},\chi_{i2},\cdots,\chi_{ij'},\cdots,\chi_{d_li}\}$ is a permutation of $\{(\xi_i-\xi_j)|g_l(\xi_i-\xi_j)\mu_i=0\}$. Let the inverse permutation of indices be denoted by 
	$\pi_i$, i.e. $\pi_i(j)=j'$, where $1 \le j' \le d_l$ and $j$ indexes $S_l^\top(i,j)=1$ for a fixed $i$. Then $Y \in \mathcal{T}$ may equivalently be written as 

	\begin{align*}
		Y &= \sum_{j'=1}^{d_l}\chi_{j'}\nu_{j'} = \sum_{j'=1}^{d_l}\sum_{i=1}^n \chi_{ij'}\mu_i \nu_{j'}\\
		  &= \sum_{i=1}^n \sum_{j'=1}^{d_l} \chi_{ij'} \nu_{j'}\mu_i \\
		  &= \sum_{i=1}^n \sum_{j:S_l^\top(i,j)=1} (\xi_i-\xi_j) \nu_{\pi_i(j)} \mu_i 
	\end{align*}

	where the set $\{\nu_{\pi_i(j)}|S_l^\top(i,j)=1\}$ is the same as  $\{\nu_1,\cdots,\nu_{d_l}\}$. Consider the polynomial ring $\mathcal{T}[z]$ and let $g_t \in \mathcal{T}[z]$. Since the expression for $g_t \in 
	\mathcal{R}[z]$ is

	\begin{align*}
		g_t(z,x) &= \sum_{i=1}^n \prod_{k:S_t(i,k)=1} (z-(\xi_i-\xi_k)) \mu_i
	\end{align*}

	Hence the expression for $g_t \in \mathcal{T}[z]$ is given by 

	\begin{align*}
		g_t(z,y,x) &= \sum_{i=1}^n \sum_{j'=1}^{d_l} \prod_{k:S_t(i,k)=1} (z-(\xi_i-\xi_k)) \nu_{j'}\mu_i
	\end{align*}

	Consider now the polynomial $g_t(z+Y,y,x) \in \mathcal{T}[z]$.

	\begin{align*}
		g_t(z+Y,y,x) &= \sum_{i=1}^n \sum_{j'=1}^{d_l} \prod_{k:S_t(i,k)=1} 
		\left(z+Y-(\xi_i-\xi_k)\right) \nu_{j'}\mu_i \\
		&= \sum_{i=1}^n \sum_{j'=1}^{d_l} \prod_{k:S_t(i,k)=1} \left(z+
		 \sum_{i_1=1}^n \sum_{j_1:S_l^\top(i_1,j_1)=1} (\xi_{i_1}-\xi_{j_1})\nu_{\pi_{i_1}(j_1)} \mu_{i_1} -(\xi_i-\xi_k)\right) \nu_{j'}\mu_i \\
		&= \sum_{i=1}^n \sum_{j:S_l^\top(i,j)=1} \prod_{k:S_t(i,k)=1} (z - (\xi_j-\xi_k)) \nu_{\pi_i(j)}\mu_i \\
		&= \sum_{i=1}^n \sum_{j:S_l(j,i)=1} \prod_{k:S_t(i,k)=1} (z-(\xi_j-\xi_k)) \nu_{\pi_i(j)}\mu_i
	\end{align*}

	Where the last couple of steps follow from the previous step because $\mu_i\mu_{i'}=0, i \neq i'$ and $\mu_i^2=\mu_i$, while $\nu_{j}\nu_{j'}=0,j\neq j'$, otherwise $\nu_{j}^2$ being just $\nu_{j}$. Note that the polynomial
	corresponding to $S_lS_t$ (denoted by $g_{lt} \in \mathcal{R}[y]$) would be of the form:

	\begin{align*}
		g_{lt}(z,x) &= \sum_{j=1}^n \prod_{\substack{
		                                            i,k: \\
							    S_l(j,i)=1 \\
						            S_t(i,k)=1}} (z-(\xi_j-\xi_k))\mu_j
	\end{align*}

	We wish to obtain $g_{lt}$ from the polynomial $g_t(z+Y,y,x)=h(z,y,x)$. This is obtained by eliminating $z,x$ from $h$ as follows. Consider the ring $\mathcal{R}'\equiv \mathbb{F}_p[y]/(f(y))$ and $\mathcal{T}'\equiv
	\mathcal{R}'[x]/(g_l(x,y))$. Consider $h(z,y,x) \in \mathcal{T}'[z]$ from which we construct the ring $\mathcal{U}\equiv \mathcal{T}'[z]/(h(z,y,x)) \equiv \mathcal{T}'[Z]$ where $Z\equiv z \pmod{h}$.
	Let $c_{\mathcal{R}'}(w,y) \in \mathcal{R}'[w]$ be the characteristic polynomial of $Z$ over the ring $\mathcal{R}'$. Then we have the following lemma:

	\begin{lemma}
		$\gcd(c_{\mathcal{R}'}(z,y),g(z,y)) = g_{lt}(z,y) \in \mathcal{R}'[z]$ where $g$ is the polynomial defined in lemma \ref{g}. In other words $g_{lt}(z,x) \in \mathcal{R}[z]$ may be obtained by merely substituting $x$
		for $y$ in this $\gcd$ polynomial.
	\end{lemma}

	\begin{proof}
		Let the primitive idempotents of $\mathcal{R}'$ over $\mathbb{F}_p$ be $\nu'_1,\cdots,\nu'_n$ ($\mathcal{R}'$ is \emph{semisimple} over $\mathbb{F}_p)$ and the primitive idempotents of $\mathcal{T}'$ over 
		$\mathcal{R}'$ be $\mu'_1,\cdots,\mu'_{d_l}$. We wish to consider $h(z,y,x)$ as a polynomial in $\mathcal{T}'[z]$; hence expressing every $\nu_{j} \mu_i$ in terms of $\mu'_j\nu'_i$'s is enough to express $h$
		as a member of this ring. We have

		\begin{align*}
			\nu_{\pi_i(j)}\mu_i &= \nu'_{j}\mu'_{i} + \sum_{r:S_l(r,i)=0} 
			\prod_{\substack{
												 q: \\
												 q \neq j \\
												 S_l(j,i)=1 \\
											 S_l(q,i)=1}} 
											 \frac{\xi_r-\xi_q}{\xi_j-\xi_q} 
											 \nu'_r\mu'_i
		\end{align*}

		where  

		\begin{align*}
			\nu_{\pi_i(j)}\mu_i &= \prod_{k:k\neq i} \frac{x-\xi_k}{\xi_i-\xi_k}\prod_{\substack{
													     q: \\
													     q \neq j\\
													     S_l(j,i)=1\\
													     S_l(q,i)=1}} \frac{y-\xi_q}{\xi_j-\xi_q}\\
			\nu'_j\mu'_i &= \prod_{q:q\neq j} \frac{y-\xi_q}{\xi_i-\xi_q} \prod_{\substack{
													k:\\
													k \neq i\\
												        S_l(j,i)=1\\
												S_l(j,k)=1}} \frac{x-\xi_k}{\xi_j-\xi_k}
		\end{align*}

		The expression for $\nu_{\pi_i(j)}{\mu_i}$ follows from evaluating it at $x=\xi_j,y=\xi_i,\xi_j \in \{\xi_1,\cdots,\xi_n\}$ and $\xi_i:S_l(\xi_j,\xi_i)=1$. This then gives us the expression for $h(z,y,x) \in 
		\mathcal{T}'[z]$ as follows

		\begin{align*}
			h(z,y,x) &= \sum_{i=1}^n \sum_{j:S_l(j,i)=1} 
			\prod_{k:S_t(i,k)=1} \left(z-\left(\xi_j-\xi_k\right)\right)
			 \left(\nu'_j + \sum_{r:S_l(r,i)=0} \alpha_{rj}\nu'_r\right)\mu'_i\\
			\alpha_{rj} &= \prod_{\substack{
							q:\\
							q\neq j\\
							S_l(j,i)=1\\
						S_l(q,i)=1}} \frac{\xi_r-\xi_q}{\xi_j-\xi_q}
		\end{align*}

		Observe that for any $1 \le i \le n$, $\nu'_j\nu'_r=0$ for $j,r:S_l(j,i)=1,S_l(r,i)=0$. Therefore we have that 

		\begin{align*}
			\gcd(c_{\mathcal{R}'}(z,y),g(z,y)) &= \sum_{j=1}^n \prod_{\substack{
											i,k:\\
											S_l(j,i)=1\\
										S_t(i,k)=1\\}} (z-(\xi_j-\xi_k))\mu'_j = g_{lt}(z,x)
		\end{align*}
	\end{proof}

	\subsubsection{The Remaining Steps}
	Now that we have the polynomial $g_{lt}$ corresponding to $S_lS_t$ we need to show that we can express the product as $\sum_{k}kD_k$ where each $D_k$ has entries from the set $\{0,1\}$. This is however easy to reproduce: 
	expressing $g_{lt}=\prod_{k}g_k^k$ where each $g_k$ is \emph{square-free} and \emph{mutually prime} gives us the polynomials corresponding to $D_k$. Further we have the following relations

	\begin{align*}
		S_l \cap S_t &\equiv \gcd(g_l,g_t) \\
		S_l \setminus S_t &\equiv \frac{g_l}{\gcd(g_l,g_t)}\\
		S_t \setminus S_l &\equiv \frac{g_t}{\gcd(g_l,g_t)}
	\end{align*}

	where we use the $\equiv$ symbol to denote the equivalence between a color matrix and its corresponding polynomial form. Note also that from lemma \ref{transposepol} that at every step of the iteration the set $\mathcal{S}$
	is closed under taking transposes, so that for every color polynomial $g_l \in \mathcal{S}$ we may also find $\bar{g}_l \in \mathcal{S}$. Hence it is possible to repeat this process till we achieve a stabilization of the
	colors. This gives rise to the following algorithm for the polynomials.

	\begin{algorithm}[H]
		\caption{2D WL for Polynomials}
		\label{polywl}
		\begin{algorithmic}
			\State $\mathcal{C} \leftarrow \mathcal{E}$
			\While{True}
			\If{\State $\forall g_i,g_j \in \mathcal{C}, g_{ij} = \prod_{l}g_{\alpha_l}^l$ for some index $\{\alpha_l\}$ of $\mathcal{C}$}
			   \State Output $\mathcal{C}$ as the final set
			\Else
			   \State $g_{lt}= \prod_{l}(h_l)^l$ where each $h_l$ is \emph{square-free} and \emph{mutually prime} to each other
			   \State $\mathcal{C} \leftarrow \mathcal{C} \cup_{l} \{h_l\}$
			   \While{$\exists g_i,g_j \in \mathcal{C},\gcd(g_i,g_j)\neq 1$}
			         \State $g=\gcd(g_i,g_j)$
				 \State $\mathcal{C} \rightarrow (\mathcal{C} \setminus \{g_i,g_j\}) \cup \{g,\frac{g_i}{g},\frac{g_j}{g}\}$
			   \EndWhile
			\EndIf
			\EndWhile
		\end{algorithmic}
	\end{algorithm}

         Note that given two polynomials $g_l$ and $g_t$, the polynomial representation for $g_{lt}$ is computable in time polynomial in the degrees of $g_l,g_t$ and $\log{p}$. Since the degree of each $g_l$ is always bounded above
	 by $n$, algorithm \ref{polywl} terminates in time polynomial in $n$ and $\log{p}$.

	 \subsection{Colors and Schemes}\label{sec:scheme}
	 Consider the final set of colors (or polynomials) $\mathcal{C}$ we obtain from algorithm \ref{polywl}. In this section we will prove some simple properties about this set $\mathcal{C}$.

	 \begin{lemma}\label{thin}
		 $2 \le |\mathcal{C}| \le n$. Further if $|\mathcal{C}|=n$ then $f$ may be factored.
	 \end{lemma}

	 \begin{proof}
		The first part follows because $|\mathcal{E}| \ge 2$, since for any $\xi_i,\xi_j$ that are roots of $f \in \mathbb{F}_p$, $\xi_i-\xi_j$ and $\xi_j-\xi_i$ differ in their \emph{$2$-Sylow} expansion in atleast one place
		- namely, the most significant bit. This is because $\xi_i-\xi_j = -1(\xi_j-\xi_i) = \eta^{2^{r-1}} (\xi_j-\xi_i)$. Since $|\mathcal{C}| \ge |\mathcal{E}| \ge 2$ the first inequality follows. The second inequality 
		is trivial since the product of all the polynomials in $\mathcal{C}$ is $f(y,x) \in \mathcal{R}[y]$ (or equivalently, $\sum_{C_l \in \mathcal{C}} C_l = J$). If $|\mathcal{C}|=n$ then the degree of any polynomial
		$g_l \in \mathcal{C}$ is $1$ which gives us an endomorphism of the roots of $f$ and hence $f$ may be factored by \cite{evdokimov}.
	\end{proof}

	\begin{lemma} \label{identity}
		If $I \notin \mathcal{C}$ then $f$ may be factored.
	\end{lemma}

	\begin{proof}
			Since the starting set contained $I$, the only way $I \notin \mathcal{C}$ would be that $I$ decomposes into two or more colors at some stage. Let $g_I(y,x)$ be the polynomial form of $I$:
			\begin{align*}
				g_I(y,x) &= \sum_{i=1}^n y \mu_i
			\end{align*}
			
			Suppose a color $C_l$ is a non-trivial decomposition of $I$, i.e. $C_l \subsetneq I$ and $C_l \neq \emptyset$. Then since $\forall 1 \le i \le n$ the degree of $g_I\mu_i=1$, $\exists 1 \le j \le n,
			deg(g_l\mu_j)=0$. Also since $C_l\neq \emptyset$,
			$\exists 1 \le i \le n, deg(g_l\mu_j)=1$. Hence the leading coefficient of $g_l$ will be a zero divisor in $\mathcal{R}$ and we get a decomposition of $f \in \mathbb{F}_p$.

		\end{proof}

		\begin{lemma}\label{regular}
		 Let $C_l$ be any color in the set $\mathcal{C}$. If $\exists C_l \in \mathcal{C},\exists 1 \le u,v \le n,\displaystyle \sum_{w=1}^n C_l(u,w) \neq \sum_{w=1}^n C_l(v,w)$ then $f$ may be factored. In other words, for
		 $f$ not to be factored each color $C_l \in \mathcal{C}$ must be regular.
	 \end{lemma}

	 \begin{proof}
		 Consider the polynomial form of such a $C_l \in \mathcal{C}$. 

		 \begin{align*}
			 g_l &= \sum_{i=1}^n \prod_{j:C_l(i,j)=1} (y-(\xi_i-\xi_j)) \mu_i
		 \end{align*}

		 Then it is easy to see that $\sum_{w=1}^n C_l(u,w) = deg(g_l\mu_u)$. If $deg(g_l\mu_u) \neq deg(g_l\mu_v)$, then the leading coefficient of $g_l$ is a zero-divisor in the algebra $\mathcal{R}$ and hence
		 we obtain a decomposition of $f \in \mathbb{F}_p$.
	 \end{proof}

	 It is possible to strengthen this lemma and prove that the colors are not just \emph{regular}, but also \emph{strongly regular}. We prove this in the following lemma.

	 \begin{lemma}\label{intersectionnum}
		 Consider any three colors $C_s,C_t,C_l \in \mathcal{C}$. Consider any $(i,j) \in C_l$. Then the cardinality of the set $\{k: (i,k) \in C_s,(k,j) \in C_t\}$ is independent of the choice of the edge $(i,j) \in C_l$.
		 This cardinality will henceforth be denoted by $a_{stl}$.
	 \end{lemma}
	
	\begin{proof}
		For any $C_s,C_t \in \mathcal{C}$ we know that the set $\mathcal{C}$ is closed under their multiplication, i.e. $C_sC_t = \sum_{C_k \in \mathcal{C}} \alpha_k C_k$ where $\alpha_k$ is a positive integer. It is easy
		to see from this expansion that $a_{stl} = \alpha_l$.
	\end{proof}

	Note that lemma \ref{regular} is a special case of lemma \ref{intersectionnum}, where the out-degree of color $C_l$ is given by $a_{ll^\top1}$ where $1$ denotes the identity color $I$. This brings us to the notion of 
	\emph{Association schemes} which we introduce below:

	\begin{definition}\label{scheme} Let $X$ be a finite, non-empty set and let $S$ be a set of relations on $X$. Then the pair $(X,S)$ form an association scheme if
		\begin{enumerate}
			\item $S$ is a partition of $X\times X$.
			\item For all $s \in S$, $s^* = \{(y,x)|(x,y) \in s\} \in S$.
			\item $I = \{(x,x)\}\in S$
			\item $\forall p,q,r \in S$ there is a number $a_{pqr}$ such that for all $(x,z) \in r$, $|\{y \in X| (x,y) \in p, (y,z) \in q\}|=a_{pqr}$.
		\end{enumerate}
	\end{definition}

	\begin{example}\label{schurian}
			An important class of association schemes called \emph{Schurian} or \emph{group case} schemes arise from the $2$-orbits of a transitive group action. Let $G$ be a group and let $H \le G$ be a subgroup.
			Then the action of $G$ on the cosets of $H$ by left multiplication is transitive. Consider the action of $G$ on $G/H \times
			G/H$ where $g(xH,yH) \rightarrow (gxH,gyH)$. Then the orbits of this action form an association scheme. To see this observe that they do partition $G/H \times G/H$. Further $I = \{(xH,xH)|x \in G\}$ is
			one of the orbits (since $G$ acts 
			transitively on the cosets of $H$). Consider a $2$-orbit $s=G(xH,yH)$ then it's transpose $G(yH,xH)$ is also a $2$-orbit. If $(xH,yH),(uH,vH)\in s$ then $\exists g \in G, (uH,vH)=(gxH,gyH)$ so that
			$(vH,uH)=(gyH,gxH)$ which implies that $(yH,xH),(vH,uH) \in s^*$. Suppose $p,q,r$ are any three of these $2$-orbits. Let $(xH,zH),(uH,vH) \in r$ so that $\exists g \in G (uH,vH)=(gxH,gyH)$. Then there is a
			bijection between $f:\{yH|(xH,yH)\in p, (yH,zH) \in q \} \rightarrow \{wH|(uH,wH) \in p,(wH,vH)\in q\}$ given by $f(yH)=gyH$. Thus the set of $2$-orbits under this action form an association scheme 
			denoted by $(G/H,G//H)$.
	\end{example}

	Let $X=\{1,\cdots,n\}$ where $deg(f)=n$. We then make the following claim:
	\begin{lemma}
		If after the algorithm \ref{polywl} $f$ remains unfactored, then the tuple $(X,\mathcal{C})$ is an association scheme.
	 \end{lemma}

	\begin{proof}
		Condition $1$ follows from the well-behavedness of the set $\mathcal{C}$ (see lemma \ref{wellbehaved}). Condition $2$ follows from lemma \ref{transposepol}. Condition $3$ follows from lemma \ref{identity}. Finally,
		condition $4$ follows because of lemma \ref{intersectionnum}.
	\end{proof}


	\subsection{Closed Subsets}\label{sec:subset}
      From the definition of schemes it follows that schemes are a generalization of groups, i.e. all groups are schemes. This can be seen in the following fashion. Consider a group $G$. For each element $g \in G$
      we define $C_g=\{(e,f)|e,f\in G,eg=f\}$. Let $\mathcal{C}_G=\cup_{g \in G}C_g$. Then the set $(G,\mathcal{C}_G)$ forms a scheme with $a_{pqr}=1$ if $pq=r$ and $0$ otherwise. If $|G|=n$ then this scheme has $n$ colors and is
      referred to as a \emph{thin scheme}. Notice that by lemma \ref{thin} \emph{thin schemes} are easily factored. 

      In this section we will talk about \emph{closed subsets} which are a generalization of subgroups to the scheme structure and prove some result about the closed subsets of the scheme $(X,\mathcal{C})$ from the previous 
      section. We will using the notation of Zieschang's book on Association Schemes \cite{zieschang}; a more detailed discussion on schemes and closed subsets can also be found there. If $(X,S)$ is a scheme and $R \subseteq S$
      and $x \in X$ then by $xR$ we mean the set $\{y|y \in X,(\exists s \in R,(x,y) \in s\}$.

      \begin{definition}
	      Let $(X,S)$ be any arbitrary scheme. A nonempty subset $R$ of $S$ is called closed if $R^*R \subseteq R$, where $R^*=\{s^*|s \in R\}$.
      \end{definition}

      Note that $I \in R^*R$ which implies that $I \in R$. Further $R^* = R^*I \subseteq R^*R \subseteq R$, which implies that $R=R^*$. Hence $RR \subseteq R$.	Seen in this fashion the choice of the name \emph{closed subsets}
      and the link with subgroups becomes clear. 

      \begin{definition}
	 If $R \subseteq S$ where $(X,S)$ is a scheme, then we define $X/R = \{xR|x \in X\}$.
 \end{definition}

 	We then have a following simple lemma.

	\begin{lemma}\label{partition}
		If $R \subseteq S$ then $R$ is closed if and only if the set $X/R$ is a partition of $X$.
	\end{lemma}

	\begin{proof}
		To see the forward direction we note that the relation $x~y \equiv y \in xR$ is an equivalence relation on $X$. Every $x \in X$ belongs to $xR$ since $I \in R$, so $x ~ x$. If $y \in xR$ then $\exists s \in R$ such 
		that $(x,y) \in s$. Since $R^* = R$, $s^* \in R$ so that $x \in yR$. If $y \in xR$ and $z \in yR$ then $z \in xR$ since $RR \subseteq R$. Hence if $R$ is closed, $X/R$ is a partition of $X$.
		The converse follows in a similar fashion: if $y \in xR$ then $x \in yR$ (since $X/R$ is a partition) so that $R^*=R$. From transitivity it follows that $RR = R^*R \subseteq R$ so that $R$ is \emph{closed}.
	\end{proof}

	\begin{definition}\label{primitivescheme}
		An association scheme $(X,S)$ is called primitive if its only closed subsets are $\{I\}$ and $S$ itself.
        \end{definition}

	From the definition it is clear that \emph{primitive schemes} are generalizations of \emph{prime groups} - groups which have no non-trivial subgroup. In the case of groups we know that groups with no proper subgroups
	are precisely the prime groups - cyclic (hence commutative) groups of prime order.	
	Let $n_R=\sum_{s \in R} a_{ss^*1}$ and $|X|=n$. Therefore since $X/R$ is a partition of $X$ we have from lemma \ref{partition}
	$n_R|n$. It is then obvious that schemes on a prime number of vertices ($|X|=p$, $p$ prime) must be \emph{primitive}. From \cite{Hanaki:2006} we also know that association schemes of prime order are commutative. It is 
	tempting to extend the analogy with groups and conjecture that \emph{primitive schemes} would always have a prime order (or atleast be commutative/cyclic). However this is \emph{false} (for an argument see Appendix
	\ref{sec:counterexample}).
	
	Consider the scheme $(X,\mathcal{C})$ from the previous section. We wish to prove that the
	problem of factoring the polynomial $f$ (and its associated scheme $(X,\mathcal{C})$) can always be reduced to factoring a polynomial of degree $\le deg(f)$ whose associated scheme is \emph{primitive}. Before that we
	need a definition and a minor lemma.

 	\begin{definition}
		Let $(X,S)$ be an arbitrary scheme and $R$ any subset of $S$. Then by $(R)$ we denote the intersection of all closed subsets of $(X,S)$ which contain $R$.
	\end{definition}

	We set $R^0 =\{I\}$ and inductively define $R^i=R^{i-1}R$. Then we have the following

	\begin{lemma}
		The set $(R)$ is the union of all sets $(R \cup R^*)^i$ where $i$ is a non-negative integer.
	\end{lemma}

	\begin{proof}
		Let $P=R \cup R^*$ and let $Q = \cup_{i \in \mathbb{Z}_+}(R^*\cup R)^i$. We wish to show $(R)=Q$. Since $(P^i)^*=(P^*)^i$, we have $(P^i)^*=P^i$ since $P^*=P$. Thus for any non-negative integers $l,m$ we have

		\begin{align*}
			(P^l)^*P^m = P^lP^m = P^{(l+m)} \subseteq Q
		\end{align*}

		Therefore $Q$ is closed. Further since $R \subseteq Q$, we must have $(R) \subseteq Q$ by definition. The converse follows because

		\begin{align*}
			\forall i \in \mathbb{Z}_+,P^i \subseteq (P^i) \subseteq (P) = (R)
		\end{align*}

		Hence $Q \subseteq (R)$.
	\end{proof}

	Note that when generating $(R)$ by taking the union of $(R^* \cup R)^i$ one only needs to go till $i \le |\mathcal{S}|$. This is because at every step the size of this set grows by at least one and it cannot grow beyond
	$|\mathcal{S}|$, hence it must stabilize for some $i \le |\mathcal{S}|$. We then come to the main lemma.

	\begin{theorem}\label{primitive}
		Let the association scheme we get from algorithm \ref{polywl} on the polynomial $f \in \mathbb{F}_p[x]$ be $(X,\mathcal{C})$. Here $|X|=deg(f)=n$. Then the problem of factoring $f$ over $\mathbb{F}_p$ may be 
		polynomial time reduced	to that of factoring a polynomial $g \in \mathbb{F}_p[x]$, $deg(g) \le deg(f)$ whose association scheme $(Y,\mathcal{S})$ is primitive.
	\end{theorem}

	\begin{proof}If the scheme $(X,\mathcal{C})$ is \emph{primitive} then we are already done. Hence, suppose that $(X,\mathcal{C})$ is not \emph{primitive}. Then we first claim that it is possible to find a non-trivial closed
		subset of $\mathcal{C}$ in time polynomial in $n$ and $\log{p}$. This is done by considering $(C_l)$ for all such $C_l \in \mathcal{C}$. Since $|\mathcal{C}| \le n$ and to construct $(C_l)$ we only need to take powers
		of $C_l$ till at most $|\mathcal{C}|$ we can construct all such $(C_l)$'s in time $O(n\log{p})$. Now suppose all of these sets are either $I,\mathcal{C}$ then we claim that $(X,\mathcal{C})$ must be \emph{primitive}.
		Suppose not, then there exists a set $\mathcal{D}$, $\{I\} \subsetneq \mathcal{D} \subsetneq \mathcal{C}$ which is closed. Let $C_l \in \mathcal{D}$ where $C_l \neq I$ (such a $C_l$ exists). Then by definition
		$(C_l) \subseteq \mathcal{D} \subsetneq \mathcal{C}$ and hence we arrive at a contradiction. Thus given $(X,\mathcal{C})$ we either determine that it is primitive, or in polynomial time construct a \emph{closed
		subset} of this scheme. Let $\mathcal{D}$ be the closed subset so obtained. Define by $g_{\mathcal{D}}$ the following polynomial:

		\begin{align*}
			g_{\mathcal{D}} &= \prod_{C_l \in \mathcal{D}} g_l = \sum_{i=1}^n g_{\mathcal{D}i} \mu_i
		\end{align*}

		Where $deg(g_{\mathcal{D}}) = n_{\mathcal{D}}$ and $g_{\mathcal{D}i} \in \mathbb{F}_p[y]$. From lemma \ref{partition} we know that $X/\mathcal{D}$ is a partition of $X$. Hence $\forall j \in i\mathcal{D}$, 
		$g_{\mathcal{D}j}=g_{\mathcal{D}i}$. If $j \notin i\mathcal{D}$ then $g_{\mathcal{D}i} \neq g_{\mathcal{D}j}$ so that $\exists 0 \le \alpha < n_{\mathcal{D}}$ such that the coefficient of $y^{\alpha}$ in
		$g_{\mathcal{D}i}$ is different from that of $g_{\mathcal{D}j}$. Suppose the coefficient of such a $y^{\alpha}$ is $h(x) \in \mathbb{F}_p[x]/(f)$. Then the resultant polynomial $Res(h(x)-z,f(x)) \in F_p[z]$ has
		at most $n/n_{\mathcal{D}}$ distinct roots, so that on making it square-free we get a polynomial $g$ over $\mathbb{F}_p$ whose degree is at most $n/n_{\mathcal{D}}$. Finding a root $\beta$ of $g$ then gives us a zero
		divisor in the algebra $\mathcal{R}$, namely $h(x)-\beta$. Also note that finding such a $y^\alpha$ and its coefficient is easy, all we need to do is ensure $h(x) \notin \mathbb{F}_p$. Note that we may assume that
		the scheme $(Y,\mathcal{S})$ formed by $g$ is \emph{primitive} or we can again reduce it in the above fashion.
	\end{proof}

\chapter{Conclusion and Future Work}

In this thesis we have generalized the approach of Gao \cite{gao2001deterministic} and Saha \cite{saha2008factoring} which has eventually led us closer to the combinatorial structure introduced by Saxena et al. in 
\cite{ivanyos2009schemes}. In their most recent work \cite{arora2012deterministic}, the authors use a structural theorem due to Hanaki and Uno to come up with a $poly(n^{\log{\log{n}}},\log{q})$ algorithm for factoring
infinitely many polynomials of prime degress over the field $\mathbb{F}_q$. Polynomials of prime degrees form association schemes of prime order, which are also \emph{primitive} (see definition \ref{primitivescheme}). Hence
any extension of such a result to the larger class of primitive schemes would by theorem \ref{primitive} lead to better algorithms for polynomials of arbitrary degree.

Another interesting angle that we have not yet explored is how some of the properties of the algebra formed by the colors obtained by algorithm \ref{polywl} relates to obtaining a factor of the polynomial $f$. For any field
$K$ we can consider the algebra $K\mathcal{C}$; if the characteristic of $K$ does not divide the degree of any color $C_i \in \mathcal{C}$ then this algebra is semisimple over $K$. The dimension of this algebra over $K$ is given
by $|\mathcal{C}| < n$ (if $|\mathcal{C}|=n$ then we can factor $f$ by lemma \ref{thin}). Then it would be interesting to see how the idempotents of this algebra over $K$ relate to the idempotents of the algebra $\mathcal{R}$ over
$\mathbb{F}_p$, if there is indeed any relation. In particular what would be relation if $K=\mathbb{F}_p$ (or an algebraic extension of $\mathbb{F}_p$)? Or if $K=\mathbb{C}$? Answering some of these questions would perhaps lead
us to a better understanding of the factoring problem and can serve as a possible direction for further work in this area.

\appendix
\chapter{}
\section{Primitiveness and Commutativity}\label{sec:counterexample}
It is tempting to extend the analogy between groups and schemes and conclude that \emph{primitive association schemes} (see definition \ref{primitivescheme}) are also commutative of prime order. The converse is certainly true. 
Also if we start with any element $s \in S$ (for a primitive scheme $(X,S)$) and close it under taking powers, we get the whole set $S$. However it is not true that a \emph{primitive scheme} should have prime order. Further, it is
also not true that a primitive scheme should be \emph{commutative}. In this appendix we will show how a non-commutative primitive scheme may be constructed; the existence of such a scheme is also enough to show that \emph{primitive
schemes} need not have prime order since \emph{prime order} schemes are always commutative due to \cite{Hanaki:2006}.

Consider an association scheme that is \emph{Schurian} (for definition see Example \ref{schurian}). Let it be $(G/H,G//H)=(X,S)$ for some group $G$ and its subgroup $H$. We have the following lemma.

\begin{lemma}\label{primitiveperm}
	Consider the image of the action of $G$ on $G/H$ as a subgroup of the symmetry group on $G/H$. Then the scheme $(G/H,G//H)=(X,S)$ is primitive iff this image is a primitive permutation group.
\end{lemma}

\begin{proof}We will prove this by contradiction. Let $X=\cup_{i} T_i$ be a non-trivial partition of $G/H=X$ that is preserved by the image of this action. By a trivial partition we mean the partition into singletons and the 
	partition consisting of the whole set. For $x,y \in X$ let $\mathcal{O}(x,y)$ denote the orbit under the action of $G$ on $X\times X$. For any $x \in X$ let $T_x$ denote the partition it belongs to. Consider the set 
	$T= \cup_{y \in T_x} \mathcal{O}(x,y)$. Then we claim that that $T$ is a non-trivial closed subset of $S$. Firstly, $T \subsetneq S$ since $\exists z \in X, z \notin T_x$ (since the partition is non-trivial) and hence
	$\mathcal{O}(x,z) \notin T$. Further $T^*=T$ since we may equivalenty represent $T = \cup_{z \in T_y} \mathcal{O}(y,z)$ and $x \in T_y$. We need to show that $TT=T$. Let $s_i,s_j \in T$. Let $s_i=\mathcal{O}(x,y)$
	and $s_j=\mathcal{O}(y,z)$ (such a representation is always possible since the action of $G$ on $X$ is \emph{transitive}). Then $y \in T_x$ and $z \in T_y$ so that $z \in T_x$ so that $\mathcal{O}(x,z) \in T$. The converse
	proceeds in a similar fashion. If $T$ is a closed subset of $S$ then it is easy to see that the set $X/T$ forms a partition of $X$ (lemma \ref{partition}) that is preserved.
\end{proof}

The \emph{O'Nan-Scott theorem} \cite{aschbacher1985maximal} gives us a complete classification of the \emph{maximal subgroups} of $\Sym{n}$, where $\Sym{n}$
denotes the symmetry group on $n$ letters. They are precisely:

\begin{enumerate}
	\item Intransitive group of the form $\Sym{k} \times \Sym{n-k}$.

	\item Imprimitive group of the form $\Sym{k} \wr \Sym{m}$ where $mk=n$ and
	 $\wr$ denotes the \emph{wreath product}.

	\item $\Sym{k} \wr \Sym{m}$ where $k^m=n$.

	\item $\AGL(d,p)$ where $p^d=n$ and $\AGL$ the \emph{affine general linear group}.

	\item $D(T,k)$ where $|T|^{k-1}=n$ and \(D\) is a diagonal group.

	\item An \emph{almost simple} group in some primitive action. (A group is \emph{almost simple} if it lies between a non-abelian simple group and its automorphism group).
\end{enumerate}

On the other hand a Schurian scheme $(G/H,G//H)$ is \emph{commutative} when $(G,H)$ form what is called the \emph{Gelfand pair} (for a proof and more discussion see \cite{wielandt1964finite}). We have due to Saxl 
\cite{saxl1981multiplicity} a classification of \emph{Gelfand pairs} $(\Sym{n},K)$. Roughly, $K$ must be contained as a subgroup of ``small index'' in one of these groups (if $n>18$).

\begin{enumerate}
	\item $\Sym{n-t}\times \Sym{t}$
		
	\item $\Sym{n/2} \wr \Sym{2}$ or $\Sym{2} \wr \Sym{n/2}$ for even $n$.

	\item $\Sym{n-5} \times \AGL(1,5)$
		
	\item $\Sym{n-6} \times \PGL(2,5)$

	\item $\Sym{n-9} \times \PGAL(2,8)$.
\end{enumerate}

This gives us the existence of \emph{primitive association schemes} that are not \emph{commutative} - for instance when $G=\Sym{n}$ for some odd $n$ and $H=\Sym{m} \wr \Sym{k}$ such that $mk=n$ and $m(k-1)>9$. This is because
$H$ is not contained in groups of the form $2,3,4$ or $5$. Further if $\Sym{m} \le \Sym{n-t}$ for some $t$, then we have $m \le n-t$ which implies $t \le n-m = m(k-1)$ so that the only way $\Sym{m} \wr \Sym{k-1}$ is contained in $\Sym{t}$ is
if $\Sym{t} = \Sym{m} \wr \Sym{k-1}$ but $k\neq 2$ ($n$ is odd) so this is not possible. Hence the scheme $(G/H,G//H)$ is \emph{primitive} but not \emph{commutative}.

\newpage
\bibliographystyle{alpha}
\addcontentsline{toc}{chapter}{Bibliography}
\bibliography{main}

\begin{thebibliography}{BVZGLJ01}

\bibitem[Agr07]{barbados}
Manindra Agrawal.
\newblock {\em Rings and Integer Lattices in Computer Science}.
\newblock Annual Workshop in Computational Complexity. 2007.

\bibitem[AIKS12]{arora2012deterministic}
M.~Arora, G.~Ivanyos, M.~Karpinski, and N.~Saxena.
\newblock Deterministic polynomial factoring and association schemes.
\newblock {\em arXiv preprint arXiv:1205.5653}, 2012.

\bibitem[AMM77]{adleman1977taking}
L.~Adleman, K.~Manders, and G.~Miller.
\newblock On taking roots in finite fields.
\newblock In {\em 18th Annual Symposium on Foundations of Computer Science},
  pages 175--178. IEEE, 1977.

\bibitem[Ank52]{ankeny1952least}
NC~Ankeny.
\newblock The least quadratic non residue.
\newblock {\em Annals of mathematics}, pages 65--72, 1952.

\bibitem[AS85]{aschbacher1985maximal}
M.~Aschbacher and L.~Scott.
\newblock Maximal subgroups of finite groups.
\newblock {\em J. Algebra}, 92(1):44--80, 1985.

\bibitem[Bab80]{babai1980isomorphism}
L.~Babai.
\newblock Isomorphism testing and symmetry of graphs.
\newblock {\em ANNALS DISCRETE MATH.}, 8:101--110, 1980.

\bibitem[Ber68]{berlekamp1}
E.R. Berlekamp.
\newblock {\em Algebraic Coding Theory}.
\newblock McGraw-Hill, New York, 1968.

\bibitem[Ber70]{ber70}
E.R. Berlekamp.
\newblock Factoring polynomials over large finite fields.
\newblock {\em Math. Comp}, 24(111):713--735, 1970.

\bibitem[Buc90]{buchmann}
J.~Buchmann.
\newblock Complexity of algorithms in algebraic number theory.
\newblock {\em Number Theory. Proc. First Conf. Canadian Number Theory
  Association}, pages 37--53, 1990.

\bibitem[BVZGLJ01]{bach2001factoring}
E.~Bach, J.~Von Zur~Gathen, and H.W. Lenstra~Jr.
\newblock Factoring polynomials over special finite fields.
\newblock {\em Finite Fields and Their Applications}, 7(1):5--28, 2001.

\bibitem[CR85]{chor}
B.~Chor and R.L. Rivest.
\newblock A knapsack type public key cryptosystem based on arithmetic in finite
  fields.
\newblock {\em IEEE Trans. Inf. Theory}, IT-34:901--909, 1985.

\bibitem[CZ81]{cz81}
D.G. Cantor and H.~Zassenhaus.
\newblock A new algorithm for factoring polynomials over finite fields.
\newblock {\em Mathematics of Computation}, pages 587--592, 1981.

\bibitem[Evd89]{evdokimov1989factoring}
S.A. Evdokimov.
\newblock Factoring a solvable polynomial over a finite field and generalized
  riemann hypothesis.
\newblock {\em Zapiski Nauchnykh Seminarov POMI}, 176:104--117, 1989.

\bibitem[Evd94]{evdokimov}
Sergei Evdokimov.
\newblock Factorization of polynomials over finite fields in subexponential
  time under grh.
\newblock In LeonardM. Adleman and Ming-Deh Huang, editors, {\em Algorithmic
  Number Theory}, volume 877 of {\em Lecture Notes in Computer Science}, pages
  209--219. Springer Berlin Heidelberg, 1994.

\bibitem[Gao01]{gao2001deterministic}
S.~Gao.
\newblock On the deterministic complexity of factoring polynomials.
\newblock {\em Journal of Symbolic Computation}, 31(1):19--36, 2001.

\bibitem[HU06]{Hanaki:2006}
Akihide Hanaki and Katsuhiro Uno.
\newblock Algebraic structure of association schemes of prime order.
\newblock {\em J. Algebraic Comb.}, 23(2):189--195, March 2006.

\bibitem[Hua85]{huang1985riemann}
MD~Huang.
\newblock Riemann hypothesis and finding roots over finite fields.
\newblock In {\em Proceedings of the seventeenth annual ACM symposium on Theory
  of computing}, pages 121--130. ACM, 1985.

\bibitem[Hua91]{huang1991generalized}
M.D.A. Huang.
\newblock Generalized riemann hypothesis and factoring polynomials over finite
  fields.
\newblock {\em Journal of Algorithms}, 12(3):464--481, 1991.

\bibitem[IKS09]{ivanyos2009schemes}
G.~Ivanyos, M.~Karpinski, and N.~Saxena.
\newblock Schemes for deterministic polynomial factoring.
\newblock In {\em Proceedings of the 2009 international symposium on Symbolic
  and algebraic computation}, pages 191--198. ACM, 2009.

\bibitem[Jr.91]{lenstra1}
H.W.~Lenstra Jr.
\newblock On the chor-rivest knapsack cryptosystem.
\newblock {\em Journal of Cryptology}, 3:149--155, 1991.

\bibitem[KS95]{ks95}
E.~Kaltofen and V.~Shoup.
\newblock Subquadratic-time factoring of polynomials over finite fields.
\newblock In {\em Proceedings of the twenty-seventh annual ACM symposium on
  Theory of computing}, pages 398--406. ACM, 1995.

\bibitem[LLL82]{lenstra2}
A.K. Lenstra, H.W. Lenstra, and L.~Lov{\'a}sz.
\newblock Factoring polynomials with rational coefficients.
\newblock {\em Mathematische Annalen}, 261(4):515--534, 1982.

\bibitem[MS88]{mignotte1988calcul}
M.~Mignotte and CP~Schnorr.
\newblock Calcul d{\'e}terministe des racines d’un polyn{\^o}me dans un corps
  fini.
\newblock {\em Comptes Rendus Acad{\'e}mie des Sciences (Paris)}, 306:467--472,
  1988.

\bibitem[Odl85]{odlyzko}
A.~Odlyzko.
\newblock Discrete logarithms and their cryptographic significance.
\newblock {\em Advances in Cryptology: Proceedings of EUROCRYPT 1984}, LNCS
  209:224--314, 1985.

\bibitem[R{\'o}n88]{ronyai1988factoring}
L.~R{\'o}nyai.
\newblock Factoring polynomials over finite fields.
\newblock {\em Journal of Algorithms}, 9(3):391--400, 1988.

\bibitem[R{\'o}n89]{ronyai1989factoring}
L.~R{\'o}nyai.
\newblock Factoring polynomials modulo special primes.
\newblock {\em Combinatorica}, 9(2):199--206, 1989.

\bibitem[R{\'o}n92]{ronyai1992galois}
L.~R{\'o}nyai.
\newblock Galois groups and factoring polynomials over finite fields.
\newblock {\em SIAM Journal on Discrete Mathematics}, 5(3):345--365, 1992.

\bibitem[Sah08]{saha2008factoring}
C.~Saha.
\newblock Factoring polynomials over finite fields using balance test.
\newblock {\em arXiv preprint arXiv:0802.2838}, 2008.

\bibitem[Sax81]{saxl1981multiplicity}
J.~Saxl.
\newblock Oin multiplicity free permutation representations.
\newblock In {\em Finite Geometries and Designs: Proceedings of the Second Isle
  of Thorns Conference 1980}, volume~49, page 337. Cambridge University Press,
  1981.

\bibitem[vL82]{vanlint}
J.H. van Lint.
\newblock {\em Introduction to Coding Theory}, volume~86 of {\em Graduate Texts
  in Mathematics}.
\newblock Springer-Verlag, New York, 1982.

\bibitem[vzG87]{von1987factoring}
J.~von~zur Gathen.
\newblock Factoring polynomials and primitive elements for special primes.
\newblock {\em Theoretical Computer Science}, 52(1):77--89, 1987.

\bibitem[VZGP01]{von2001factoring}
J.~Von Zur~Gathen and D.~Panario.
\newblock Factoring polynomials over finite fields: A survey.
\newblock {\em Journal of Symbolic Computation}, 31(1):3--17, 2001.

\bibitem[VZGS92]{gs92}
J.~Von Zur~Gathen and V.~Shoup.
\newblock Computing frobenius maps and factoring polynomials.
\newblock {\em Computational complexity}, 2(3):187--224, 1992.

\bibitem[WB64]{wielandt1964finite}
H.~Wielandt and R.~Bercov.
\newblock {\em Finite permutation groups}, volume~10.
\newblock Academic Press New York, 1964.

\bibitem[WL68]{weisfeiler1968reduction}
B.J. Weisfeiler and AA~Leman.
\newblock Reduction of a graph to a canonical form and an algebra which appears
  in the process.
\newblock {\em NTI, Ser}, 2(9):12--16, 1968.

\bibitem[WS77]{macwilliams}
F.J.~Mac Williams and N.J.A Sloane.
\newblock {\em The Theory of Error-correcting Codes}, volume~16.
\newblock Holland, 1977.

\bibitem[Zie05]{zieschang}
P.H. Zieschang.
\newblock {\em Theory of Association Schemes}.
\newblock Springer, 2005.

\end{thebibliography}


\end{document}